\documentclass[a4paper,11pt]{article} 

\usepackage[labelfont=bf,labelsep=period]{caption}
\usepackage[nodayofweek]{datetime}
\usepackage{enumitem}
\usepackage{float}
\usepackage[margin=2.5cm]{geometry}
\usepackage{graphicx}
\usepackage{hyperref}
\usepackage[utf8]{inputenc}
\usepackage{numbertabbing}
\usepackage{times}
\usepackage{url}
\usepackage[dvipsnames]{xcolor}
\usepackage{xspace}

\usepackage{array}
\usepackage{subcaption}
\usepackage{tikz}
\usepackage{amsfonts}
\usepackage{booktabs}
\usepackage{makecell}

\usepackage{amsmath} 
\allowdisplaybreaks[2]          
\usepackage{amssymb} 
\usepackage{amsthm} 
\newtheoremstyle{plain-boldhead}
  {\topsep}
  {\topsep}
  {\itshape}
  {}
  {\bfseries}
  {.}
  { }
  {\thmname{#1}\thmnumber{ #2}\thmnote{ (\bfseries #3)}}
\newtheoremstyle{definition-boldhead}
  {\topsep}
  {\topsep}
  {\normalfont}
  {}
  {\bfseries}
  {.}
  { }
  {\thmname{#1}\thmnumber{ #2}\thmnote{ (\bfseries #3)}}
\theoremstyle{plain-boldhead}
\newtheorem{theorem}{Theorem}

\newtheorem{lemma}[theorem]{Lemma}
\newtheorem{corollary}[theorem]{Corollary}

\theoremstyle{definition-boldhead}
\newtheorem{definition}{Definition}

\newdateformat{simple}{\THEDAY\ \monthname[\THEMONTH]\ \THEYEAR}
\simple

\floatstyle{ruled}
\newfloat{algo}{htbp}{algo}
\floatname{algo}{Algorithm}

\def \ifempty#1{\def\temp{#1} \ifx\temp\empty }

\newcommand{\str}[1]{\textsc{#1}}
\newcommand{\var}[1]{\textit{#1}}
\newcommand{\op}[1]{\textsl{#1}}
\newcommand{\msg}[2]{\ensuremath{\ifempty{#2} [\str{#1}] \else [\str{#1}, {#2}] \fi}}

\newcommand{\etal}{\emph{et al.}}

\newcommand{\CC}{\ensuremath{\mathcal{C}}\xspace}

\newcommand{\CE}{\ensuremath{\mathcal{E}}\xspace}

\newcommand{\CL}{\ensuremath{\mathcal{L}}\xspace}
\newcommand{\CM}{\ensuremath{\mathcal{M}}\xspace}

\newcommand{\CP}{\ensuremath{\mathcal{P}}\xspace}
\newcommand{\CQ}{\ensuremath{\mathcal{Q}}\xspace}

\newcommand{\CS}{\ensuremath{\mathcal{S}}\xspace}

\newcommand{\CV}{\ensuremath{\mathcal{V}}\xspace}

\usepackage{tikz}
\usetikzlibrary{arrows.meta}
\usetikzlibrary{shapes.geometric}

\newcommand{\setup}[1]{\ensuremath{\op{VE.Setup(\ensuremath{#1})}\xspace}}
\newcommand{\encapsulate}[1]{\ensuremath{\op{VE.Encapsulate}(\ensuremath{#1})}\xspace}
\newcommand{\validate}[2]{\ensuremath{\op{VE.Validate}_{\ensuremath{#1}}(\ensuremath{#2})}\xspace}
\newcommand{\verify}[1]{\ensuremath{\op{VE.Verify}(\ensuremath{#1})}\xspace}
\newcommand{\extract}[2]{\ensuremath{\op{VE.Extract}_{\ensuremath{#1}}(\ensuremath{#2})}\xspace}
\newcommand{\reconstruct}[1]{\ensuremath{\op{VE.Reconstruct(\ensuremath{#1})}\xspace}}
\newcommand{\encaps}{\ensuremath{\var{e}\xspace}}
\newcommand{\share}{\ensuremath{\var{s}\xspace}}

\usepackage[backend=bibtex,style=numeric,giveninits=true,maxbibnames=99]{biblatex}
\renewbibmacro*{doi+eprint+url}{%
  \iftoggle{bbx:url}
  {\iffieldundef{doi}{\usebibmacro{url}}{}}
  {}%
  \newunit\newblock
  \iftoggle{bbx:eprint}
  {\usebibmacro{eprint}}
  {}%
  \newunit\newblock
  \iftoggle{bbx:doi}
  {\printfield{doi}}
  {}
}

\begin{document}

\title{\bf Byzantine Reliable Broadcast with Causal Ordering}

\author{Mariarosaria Barbaraci\footnotemark[1]\\
University of Bern\\
\url{mariarosaria.barbaraci@unibe.ch}
  \and Christian Cachin\footnotemark[1]\\
  University of Bern\\
  \url{christian.cachin@unibe.ch}
}

\footnotetext[1]{Institute of Computer Science, University of Bern,
   Neubr\"{u}ckstrasse 10, 3012 CH-Bern, Switzerland.}

\date{\today} 

\maketitle

\begin{abstract}\noindent
  Reliable and total-order broadcasts in the Byzantine-fault model are well
  studied, but adding causal order has received comparatively little
  attention, largely due to the complexity that stems from actions of
  Byzantine processes. Existing solutions almost exclusively build causal
  ordering on top of total-order broadcast. The combination of causal order
  with reliable broadcast remains rare, and the few solutions that exist adopt
  the classical definition of causality based on events occurring at
  individual processes (the happened-before relation). We show this definition
  is not sufficient to enforce causal ordering among broadcast messages:
  Byzantine processes can lie about, omit, and forge dependency information
  and thereby violate the causal order among self-reported events. Such
  manipulations remain indistinguishable from correct behavior to any single
  observer. We demonstrate the issue and its consequences concretely via a
  front-running attack that violates causality in reliable broadcast, but that
  cannot be captured through the existing definitions.

  To close this gap, we extend the notion of reliable broadcast to externalize
  local potential knowledge.  We use this to formalize the first
  \emph{complete} definition of causal message ordering in reliable broadcast
  under Byzantine faults. Unlike the classical formalization, this notion is
  grounded in the joint observations of a sufficiently large group of correct
  processes rather than a process's own view. Building on this definition, we
  characterize the properties of a Byzantine reliable broadcast channel that
  guarantees causal ordering. We then present an efficient protocol that
  satisfies these properties: it is resilient to the optimal number of
  $f < n/3$ Byzantine faults and for one instance that broadcasts payload
  message~$m$, it has bit complexity $O(n^2(|m| + \lambda + n))$, where
  $\lambda$ denotes the maximal size of a unique (cryptographic) label
  for~$m$.  Finally, we analyze the properties of this protocol and prove that
  it achieves Byzantine reliable broadcast with causal ordering.

\end{abstract}

\section{Introduction}
\label{sec:intro}

For implementing distributed services in potentially adversarial environments, causal ordering is essential for applications where the semantic correctness of operations depends on preserving event dependencies. Systems such as decentralized finance, collaboration platforms, and critical infrastructure are particularly vulnerable to adversaries who exploit inconsistent ordering to manipulate state, deceive users, or disrupt operations. By ensuring that correct processes observe causally related events in an order that respects these dependencies, causal ordering strengthens both reliability and security in distributed systems.
 
Many solutions~\cite{DBLP:journals/toplas/ReiterB94,DBLP:conf/crypto/CachinKPS01,DBLP:conf/dsn/DuanRZ17,DBLP:conf/crypto/Kelkar0GJ20,DBLP:conf/ccs/KelkarDLJK23,DBLP:conf/fc/CachinMSZ22,DBLP:conf/osdi/ZhangSCZA20} address this problem in \emph{total-order broadcast} or \emph{consensus} protocols and enforce a causal order on payload messages delivered to an application.
They occupy a prominent position in the blockchain space because they prevent \emph{front-running attacks} on consensus protocols and on-chain decentralized finance~\cite{DBLP:conf/sp/DaianGKLZBBJ20,DBLP:conf/uss/TorresCS21}.

In particular, Reiter and Birman~\cite{DBLP:journals/toplas/ReiterB94} introduced the notion of \emph{input causality} for total-order broadcast, according to which a malicious participant must not be able to induce delivery of a \emph{payload message}~$m'$ after observing a payload message~$m$ that has not yet been delivered. This may arise when Byzantine participants observe $m$ in \emph{low-level messages} exchanged during the execution of the protocol. Cachin~\etal~\cite{DBLP:conf/crypto/CachinKPS01} formalize this notion as \emph{secure causal atomic broadcast}, and Duan~\etal~\cite{DBLP:conf/dsn/DuanRZ17} revisit it by proposing new implementations. These formalizations of causal message order capture the intended restrictions in the Byzantine model in connection with total-ordered delivery of the payload messages, and all implementations rely on consensus. However, they fail to provide an independent definition of causality in this setting. 

As decentralized systems become increasingly large-scale and geographically distributed, the demand for \emph{latency-efficient} and \emph{scalable} solutions continues to grow.
And enforcing a total order on all messages introduces substantial overhead, particularly in Byzantine environments where an adversary may delay, reorder, or equivocate messages.

It has also been recognized that for many practical applications, total ordering is unnecessarily strong and too costly~\cite{DBLP:conf/sosp/MoraruAK13,DBLP:conf/opodis/RyabininGS25}, especially also for BFT state-machine replication and in the domain of cryptocurrencies~\cite{DBLP:journals/dc/DickersonGHK20,DBLP:journals/dc/GuerraouiKMPS22,DBLP:conf/dsn/CollinsGKKMPPST20,DBLP:conf/icdcs/AlposCMZ21}. Correctness often depends \emph{only} on preserving the order between causally related messages, while concurrent messages with independent requests may be processed in different orders without affecting system semantics.
By relaxing the total order of consensus to reliable broadcast, systems can process messages concurrently with the assurance that eventually all processes will observe all messages. This reduces synchronization cost, improves throughput and responsiveness, and leads to more scalable platforms.
Practical systems have recently been proposed for the blockchain space that follow this pattern~\cite{DBLP:conf/eurosys/LinFZ025,DBLP:journals/corr/abs-2501-06531,DBLP:journals/corr/abs-2506-01885}.

Even though it is well-known that total-order message delivery and
causal-order message delivery are orthogonal properties~\cite{HadzilacosT93},
formulating and implementing reliable broadcast with causal ordering in the
Byzantine model introduces significant challenges. In particular, malicious
participants may attempt to fabricate, omit, or reorder causal dependencies in
ways that are difficult to detect. In systems tolerating only benign faults,
such as crashes, the sender of a message attaches information about past
events to it, thereby \emph{self-reporting} its dependencies. This approach
falls short here because a Byzantine sender may equivocate or intentionally
manipulate dependency information to violate causality among the correct processes.
 
In this work, we study causal ordering of messages for Byzantine broadcast \emph{without} total order, in particular for \emph{reliable broadcast}. We revisit existing notions of causal message ordering in the Byzantine model~\cite{DBLP:journals/tcs/AuvolatFRT21,DBLP:journals/tpds/MisraK24} and recognize that existing definitions do not account for possible Byzantine behavior and effects that such behavior may have, as they still rely on the traditional definition of causality based on events observed at single processes (Lamport's happened-before relation~\cite{DBLP:journals/cacm/Lamport78}). To address this gap, we first motivate and introduce an extended abstraction of reliable broadcast and then formalize the first \emph{complete} notion of causal order among payload messages for the Byzantine model. It is grounded in observations made by \emph{sufficiently many correct} processes, rather than a single process's view. However, it is defined in such a way to mirror its crash-fault counterpart. 
We then define a \emph{reliable broadcast channel} in the Byzantine model with causal-order message delivery.

Furthermore, we provide a protocol of a broadcast channel that extends Bracha's~\cite{DBLP:journals/iandc/Bracha87} reliable broadcast to many concurrent instances that respect causal ordering. In one protocol instance, the sender first hides the payload message cryptographically using a \emph{verifiable encapsulation (VE)} primitive, and the algorithm then proceeds in five rounds: send, echo, ready, schedule, and reconstruction. At the outset of the protocol, processes do not observe the message and agree on a label that VE has associated with the message. A second key mechanism is the use of vector clocks. Each process maintains in a local \emph{scheduled vector} a list of how many labels per process it has already scheduled. This vector is shared with other processes every time an echo message is being sent. Then, through a \emph{causal barrier} condition on echo/ready validation, a process accepts an echo/ready message only when the attached vector clock, tracking causal dependencies, is consistent with its local scheduled vector. Consequently, the protocol ensures that labels are scheduled in a causality-respecting order. Once a label is scheduled, it is placed in an ordered queue; after enough processes have scheduled some label, the processes start the reconstruction round, recover the payload, and deliver it in that scheduling order.

To complement the above discussion and further motivate the goal of this work, consider how a front-running attack on a replicated service may arise without total-order guarantees.
Let $\CP = \{p_1, p_2, p_3, \bar{p}_4\}$ be a set of four processes, where $\bar{p}_4$ is Byzantine and controlled by an adversary that also controls the network. In the execution shown in Figure~\ref{fig:frontrun}, process $p_1$ broadcasts a message $m$. The solid arrows abstract an execution of Byzantine reliable broadcast~\cite{DBLP:books/daglib/0025983}, where the dot at the sender denotes the broadcast event and the arrowhead denotes message delivery. For clarity, the dashed arrow explicitly represents the \emph{first} message sent from $p_1$ to $\bar{p}_4$ during the protocol that contains $m$ or lets $\bar{p}_4$ obtain~$m$.
Upon learning $m$, the adversary gains the ability to inject a message $m'$ on behalf of $\bar{p}_4$ before some correct processes deliver $m$. By selectively delaying the delivery of $m$, the adversary can cause some correct processes to deliver $m'$ before $m$, while omitting the existing dependency between the two messages. As a result, the adversary effectively hides the causal relation between $m$ and $m'$ that would exist under a global view of the execution.
The attack is subtle because it makes $m$ and $m'$ appear concurrent from the perspective of correct processes. Consequently, existing reliable broadcast protocols aiming at preserving causal ordering may fail to detect this violation. The reason is that the attack does not explicitly contradict the dependency information reported by Byzantine participants, which in this case can be the empty set or arbitrarily chosen. From the viewpoint of correct processes~$p_2$ and $p_3$, the causal dependencies associated with $m'$ are simply missing, even when a hidden causal relation to $m$ exists under a global view of the execution.

\begin{figure}[th]
  \centering
  \begin{tikzpicture}[>=Stealth,thick]
    
    \draw (0,0)  node[left]{$p_1$} -- (8,0)   node[pos=0.8, above]{$p_1$ delivers $m, m', \dots$};
    \draw (0,-1) node[left]{$p_2$} -- (8,-1)  node[pos=0.8, above]{$p_2$ delivers $m', m, \dots$};
    \draw (0,-2) node[left]{$p_3$} -- (8,-2)  node[pos=0.8, above]{$p_3$ delivers $m', m, \dots$};
    \draw (0,-3) node[left]{$\bar{p}_4$} -- (8,-3)  node[pos=0.8, above]{$\bar{p}_4$ delivers $m', m, \dots$};
    
    \filldraw[black] (1,0) circle (3pt) node[above]{$m$};
    \filldraw[black] (2,-3) circle (3pt) node[below]{$m'$}; 
    \draw[->, dashed] (1,0)  -- (1.5,-3); 
    \draw[->, bend left=30] (1,0)  to (2, 0);
    \draw[->] (2,-3) -- (3.5,-2);
    \draw[->] (2,-3) -- (3,-1);
    \draw[->] (2,-3) -- (2.5,0);
    \draw[->, bend left=30] (2,-3)  to (3.5, -3);
    \draw[->] (1,0)  -- (3.5,-1); 
    \draw[->] (1,0)  -- (4,-2); 
    \draw[->] (1,0)  -- (4,-3);

  \end{tikzpicture}
  \caption{Process $\bar{p}_4$ front-runs message $m$ at $p_2$ and $p_3$.}
  \label{fig:frontrun}
\end{figure}
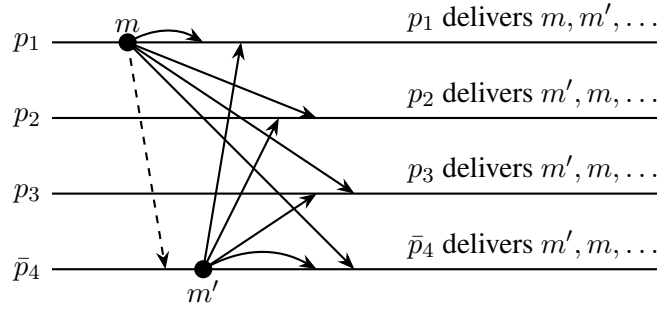

Even though almost all cryptocurrencies and blockchain networks today deliver their requests in total order among their validators, one can conceive application services for which reliable request delivery is sufficient, assuming that it also respects causal order. A name registry, for instance, may record names and assign exclusive ownership over them to clients. Names are from a large domain, hence, the service can be implemented so that different names are registered concurrently in order to increase throughput. Only when multiple clients try to register the same name, the registry has to invoke consensus and order these requests. If the communication for the registry uses Byzantine reliable broadcast, an ``interesting'' name contained in a client request may be stolen by a malicious validator process, as shown in the front-running scenario. Causally ordering the requests prevents this attack.

\paragraph{Organization.}
The rest of the paper is organized as follows. In Section~\ref{sec:background} we introduce the system model and the necessary background. In Section~\ref{sec:causality} we present our new definitions of broadcast and causality in the Byzantine setting, and we discuss the security properties required by a reliable broadcast primitive to guarantee causal ordering. In Section~\ref{sec:protocol} we present a protocol that implements the proposed abstraction and in Section~\ref{sec:analysis} show that it satisfies the required properties.
Finally, in Section~\ref{sec:relwork} we review related work. Section~\ref{sec:conclusion} concludes the paper.
Moreover, in Appendix~\ref{app:secure_VE} we discuss secure implementations of VE.

\section{Preliminaries}
\label{sec:background}

\subsection{Distributed system model}

We consider a standard distributed system consisting of a set $\CP = \{p_1, p_2, \ldots, p_n\}$ of $n$ processes that run local computations and communicate by exchanging messages over a network in order to perform a distributed protocol~\cite{DBLP:books/daglib/0025983}. Processes communicate through \emph{reliable} and \emph{authenticated} point-to-point channels. We refer to a \emph{crash-fault} model when a process may stop executing local computations and sending messages to other processes. Conversely, we refer to a \emph{Byzantine-fault} model when a process may deviate in arbitrary ways from the protocol. The failure model sets a lower bound for the size of the set of \emph{faulty} processes: with $f$ faulty processes, one requires $n > 2f$ for crash faults and $n > 3f$ for Byzantine faults. We call processes that follow the protocol \emph{correct}.

\paragraph{Timing assumptions.}
The network is asynchronous: there is no \emph{global} notion, no assumption on the time it takes for messages to be delivered, and no bound on local computation time. Each process can only rely on the events it observes, which arise either from local computation or from messages received from other processes. Following the notions introduced by Lamport~\cite{DBLP:journals/cacm/Lamport78}, every process may keep track of a time by incrementing a \emph{logical clock} that counts events.
Logical clocks naturally capture a cause-effect relation of actual causal influence in this setting. An \emph{execution} of a distributed protocol can be seen as a sequence of events visible only to a global observer. The \emph{happened-before} relation captures the potential causal precedence between events occurring at the same process (trivially) or between a reception and subsequent transmission event. In terms of logical clocks, an event happens before another if and only if the logical clock attached to it is strictly smaller.

\subsection{Broadcast abstraction}

A fundamental primitive for \emph{group communication} is the \emph{broadcast} abstraction~\cite[Sec.~3]{DBLP:books/daglib/0025983}. It allows to extend the traditional \emph{client-server} interaction to a set of processes~\cite{DBLP:journals/csur/ChocklerKV01}. 
Since multiple payload messages may be disseminated like this, the abstraction operates like a \emph{broadcast channel}. 

More formally, given a set of $n$ processes~$\CP$, a broadcast abstraction allows a sender process $p_s \in \CP$ to disseminate a payload message~$m$ to all processes in $\CP$. This primitive is characterized by a $\op{broadcast}$ event that occurs only at the sender process. We make the standard assumption that every correct process broadcasts a particular message only once.
Every process receives or \emph{delivers} a payload through a $\op{deliver}$ event. The abstraction ensures \emph{validity} in the sense that every message broadcast by a correct sender is eventually delivered; \emph{no duplication}, that no message is delivered twice; \emph{integrity}, ensuring that every delivered message was actually broadcast; and finally \emph{agreement}, the all-or-nothing property of message delivery. The latter means that if a message~$m$ is delivered by some correct process, then every correct process eventually delivers~$m$.

\paragraph{Causal broadcast.}
As protocols become more complex in the context of group communication with crash faults, the happened-before relation can be extended to capture causal dependencies between broadcast and deliver events, specifically by considering the order in which processes broadcast payload messages and how those payloads are delivered \emph{by each process} in the system.
Thus, one can extend the happened-before concept to the broadcast channel~\cite{DBLP:journals/cacm/Lamport78,DBLP:conf/sosp/BirmanJ87}, where events represent \emph{broadcast} or \emph{deliver} events of messages exchanged by the processes in the system, as follows.

\begin{definition}[Causal-order relation with crash faults]
  \label{def:causal:relation}
Given two messages~$m_1$ and~$m_2$, we say that~$m_1$ \emph{causally precedes}~$m_2$, and write $m_1 \prec m_2$, if either of the following conditions hold:
\begin{enumerate}
\item A process broadcasts~$m_1$ and then broadcasts~$m_2$.
\item A process delivers~$m_1$ and then broadcasts~$m_2$.
\item There exists a message~$m'$ such that $m_1 \prec m'$ and $m' \prec m_2$.
\end{enumerate}
If~$m_1 \not\prec m_2$ and~$m_2 \not\prec m_1$, we say that the two messages are \emph{concurrent}.
\end{definition}

The causal-order relation is used to add the property of \emph{causal delivery} to reliable broadcast, which preserves \emph{causal ordering} among the delivered messages. It ensures that when a message $m_1$ causally precedes another message $m_2$, then no process delivers $m_2$ unless it has already delivered $m_1$.

A reliable broadcast that additionally satisfies causal delivery realizes a \emph{causal broadcast} protocol under \emph{crash faults}, which is a standard notion in the distributed-computing literature~\cite{HadzilacosT93, DBLP:books/daglib/0017536}.
Causal order actually strengthens FIFO order, which restricts the delivery order only for those payload messages that have been broadcast by the same process. It is also widely understood that imposing causal order on reliable broadcast is orthogonal to requiring total order for all delivered payloads~\cite{HadzilacosT93}.

Practical implementations of causal-order broadcast for systems subject to crash faults rely on an underlying reliable broadcast mechanism for communication, together with a data structure that records the causal past of each message.
This can either take the form of a \emph{vector clock} that consists of one sequence number per sender that indicates how many payload messages from that sender have already been delivered in the context where the sender has previously broadcast the particular payload message.
The role of the vector clock is twofold: (1) a sender process includes its causal past when broadcasting a message; (2) a receiving process stores a message~$m$, together with its associated vector clock $W$, in a pending queue until all the messages reported in $W$ have been delivered.
Alternatively, and only as a conceptual solution, one might also add the complete causal past to every payload message during broadcast~\cite{HadzilacosT93, DBLP:books/daglib/0017536}.

\subsection{Byzantine reliable broadcast}
\label{sec:brb}

The existing protocols for reliable broadcast with crash faults rely on
correct processes to retransmit low-level messages and assume that all data
reported by other processes is accurate and represents the true state at
those processes. In the Byzantine model, however, this may not hold.  Instead
one has to invoke a different approach, since protocols can only rely on
actions of correct processes.  In particular relies on the view of a majority
of correct processes or a \emph{quorum} to consider an information reliable. A
\emph{Byzantine quorum} tolerating $f$ faulty processes is a set of more than
$\frac{n+f}{2}$ processes~\cite{DBLP:journals/dc/MalkhiR98}. Any two such
quorums always overlap in at least one correct process.

The notion of \emph{Byzantine reliable broadcast} is traditionally defined for
a single protocol instance where agreement is reached on delivering one
payload message. It can be extended modularly to broadcasting multiple
messages in the sense of a \emph{broadcast channel} by collecting together
many instances, of which each is identified uniquely~\cite[Sec.~3.12]{DBLP:books/daglib/0025983}. 

More precisely, each attempt to broadcast a payload message by a particular sender is associated with a \emph{label} from a global label space~$\CL$ that is partitioned by sender into disjoint subsets~$\CL_i$ for $p_i \in \CP$. When process~$p_i$ broadcasts a payload message $m$ in the message space $\CM$, the channel assigns a fresh label $\ell\in\CL_i$, and the instance is identified by $\ell$ alone, with the sender implicitly known through the partition of the label space. A partial map $\mu:\CL \rightharpoonup \CM$ represents this binding on issued labels, so that each defined label corresponds to exactly one attempt to broadcast a payload message. A correct process broadcasts a message~$m$ through a $\op{ch-broadcast($m$)}$ event. The abstraction then selects a label~$\ell$ and reports $\ell$ in a $\op{ch-deliver}(\ell, m)$ event to indicate that the payload message is received.

\begin{definition}[Byzantine reliable broadcast channel]
  \label{def:brb:channel}
  A \emph{Byzantine reliable broadcast channel (BRCH)} allows processes to
  broadcast payload messages and delivers them such that it satisfies the
  following properties for each label~$\ell$:
\begin{description}
\item[\textit{Validity}:] If a correct process~$p_s$ broadcasts a message~$m$,
  then every correct process eventually delivers~$m$ with associated
  label~$\ell$, such that $\ell \in \CL_s$.
\item[\textit{No duplication}:] For every label~$\ell$, every correct process
  delivers at most one message with label~$\ell$.
\item[\textit{Integrity}:] If some correct process delivers a message~$m$ with
  label~$\ell$, where $\ell \in \CL_s$ for some process~$p_s$ and process~$p_s$
  is correct, then~$m$ was previously broadcast by~$p_s$.
\item[\textit{Agreement}:] If some correct process delivers a message~$m$ with
  label~$\ell$, then every correct process eventually delivers message~$m$
  with label~$\ell$.
\end{description}
\end{definition}

It is worth noting that the agreement property may be decomposed in two
orthogonal ones that hold for each label: First, a safety property, \emph{consistency}, which
enforces that if two correct processes deliver a message with some
label~$\ell$, then it is the same message~$m$ that every other correct process
delivers with~$\ell$, and second, a liveness property called \emph{totality},
which captures the all-or-nothing requirement, that if one correct delivers
some message with label~$\ell$, then all correct processes eventually do
so~\cite{DBLP:books/daglib/0025983}.

Implementations of Byzantine reliable broadcast channel are derived directly
from the celebrated protocol of Bracha~\cite{DBLP:journals/iandc/Bracha87} for
broadcasting one single payload with an associated label~$\ell$.  The
broadcast channel runs one instance of it for each of the $n$ processes that
may act as senders concurrently.  When one instance terminates, it
starts the next one for this sender. This naturally defines also per-sender
sequence numbers for payload messages on the broadcast channel.

We briefly describe here one instance of Bracha's protocol.  It proceeds in three
rounds of communication: In the \str{send} round, the sender~$p_s$ sends the
message to all processes $p_i \in \CP$.  In the \str{echo} round, a correct
process $p_i$ retransmits the message from the sender $p_s$ to all.  In the
\str{ready} round, a correct process $p_i$, upon collecting a Byzantine quorum
\str{echo} messages, retransmits again the payload to all.  Upon collecting
more than $f$ \str{ready} messages, a correct process that hasn't still
received the quorum of \str{echos}, skips ahead and sends a \str{ready}
message in what is called the amplification step.  (It relies on the
fact that at least one correct process has received a quorum of \str{echo}
messages for the payload.)  And upon collecting more than $2f$ \str{ready}
messages, a correct process delivers the contained payload.  (This can be done
since there are more than $f$ correct processes that entered the \str{ready}
round and will eventually bring all correct processes to deliver.)

\section{Causal ordering in the Byzantine setting}
\label{sec:causality}

The traditional definition of causality according to
Definition~\ref{def:causal:relation} is expressed through events occurring at
all processes, including the faulty ones.  This does not extend to the most
useful notion of causal ordering for Byzantine reliable broadcast because it
lacks information on events that occur at Byzantine processes.  This section
discusses the issue and presents a solution by extending the interface of a
broadcast primitive.

\subsection{Limitations of existing approaches}

There are two fundamental issues that make the traditional notion of causal
order problematic in the Byzantine model.

\paragraph{Events for defining causality.}
The first problem derives from the lack of information on events that occur at
Byzantine processes.  Any formal notion in the Byzantine model must be stated
in terms of events occurring at correct processes, as no assumptions can be
made on the behavior of the adversary and the processes under its control.  In
other words, one cannot express that a faulty process may have ``delivered''
some payload message~$m$ that subsequently led it to ``broadcast'' a
message~$m'$ that would then causally depend on~$m$.  The adversary may skip
any kind of ``broadcasting behavior'' for~$m'$ as long as its actions make
some correct process deliver~$m'$.

One of the most advanced definitions of causal reliable broadcast with
Byzantine faults, by Auvolat~\etal~\cite{DBLP:journals/tcs/AuvolatFRT21},
indeed starts from Definition~\ref{def:causal:relation}, but \emph{restricts}
its second condition to events occurring at \emph{correct processes}. This
modification results in a weaker notion than the one we envisage here.  More
precisely, Auvolat~\etal's \emph{Byzantine causal relation} orders all payload
messages from the \emph{same} sender, whether it is faulty or correct, with
respect to each other and also considers that any payload message delivered by
a \emph{correct} process causally precedes all payloads broadcast subsequently
by that process.  (Naturally, the causal order includes also the transitive
hull.)

In particular, this means that every correct process delivers the payload
messages from one sender in the order imposed by the sender (i.e., FIFO
order).  It also implies that the local order at a correct process
among all messages it delivers and subsequently broadcasts is maintained
by the causal order.  However, the notion does not capture situations
in which such a causal influence occurs through the actions of
faulty processes, such as the example shown in the introduction
(Fig.~\ref{fig:frontrun}).

\paragraph{Low-level messages revealing information about payloads.}
The second difficulty in the Byzantine model concerns actions available to
faulty processes when they exchange low-level protocol messages.
When a correct process broadcasts a payload message~$m$, it must communicate
$m$ to the other processes through low-level messages.  Unless cryptographic
techniques are used in the protocol to hide~$m$, an adversary controlling the
network has the ability to observe information about $m$ in low-level messages
as soon as a $m$ has been broadcast.  Such information can then be exploited
by the adversary even if $m$ is never delivered by any faulty process.

Many existing notions of causal-order broadcast in the Byzantine model have
addressed \emph{only total-order message delivery} and solved the above problem by
adding confidentiality for payload messages through cryptographic means up to
a certain point in time.  This approach was pioneered by Reiter and
Birman~\cite{DBLP:journals/toplas/ReiterB94} and later extended and refined by
many
others~\cite{DBLP:conf/crypto/CachinKPS01,DBLP:conf/dsn/DuanRZ17,DBLP:conf/tokenomics/MalkhiS22}.
Their common theme are cryptographic notions of confidentiality that rely on
\emph{total-order delivery} of payloads, i.e., that all correct processes
deliver the payload messages in the same sequence.

Implementations have either used a \emph{commit-reveal} strategy, which
requires the sender to initially disseminate a commitment or sharing of the
input message, and once an order has been decided, they reveal the message.
Notably this solution fails in the presence of faulty senders, as it requires
the sender to open the commitment in a second stage and speak twice.

Other implementations have \emph{encrypted} the payloads message such that set
of processes may jointly decrypt it afterward, typically through threshold
cryptography that requires correct processes to collaborate for
decryption.

Both techniques fail in our setting that does not assume a
total order. In particular, the adversary can influence the local delivery
order at correct processes after the payload has been revealed or decrypted.
This may contradict the desired notion of causality.  The inherent problem
lies in the lack of synchronization among the correct processes for revealing
or decrypting payloads since this may occur concurrently.

\subsection{Stronger causal order for Byzantine reliable broadcast}
\label{sec:causality_definition}
  
To prevent the front-running attack, one should postpone the moment in which the adversary learns about the payload in message~$m$ to a later point after the \op{ch-broadcast} event. Moreover, correct processes must be able to ``observe'' a \emph{potential causal} dependency for a message $m$ before its delivery. 
We therefore formally define \emph{intermediate} events, which occur between \op{ch-broadcast} and \op{ch-deliver}. 

In the context of a broadcast channel, we leverage the presence of a label~$\ell$ associated with message~$m$ to define (1) when $\ell$ is first observed and (2) when the message~$m$ associated with $\ell$ can safely be disclosed. 
Crucially, because a correct process cannot ``trust'' the legitimacy of a label~$\ell$ the first time it observes it, it needs to know when a particular message referenced by the associated label~$\ell$ is being considered for delivery by sufficiently many processes.
We capture these occurrences at a correct process~$p_i$ like this: we say a correct process \op{ch-acknowledges} the presence of a label~$\ell$ when it first observes $\ell$, and we say it \op{ch-schedules} a label~$\ell$ when enough processes agree on the existence of such $\ell$, and, consequently, on the existence of the associated message~$m$. 

These intermediate events may be exposed by an enhanced broadcast abstraction. Notably, we want to be able to properly define what happens in the system, even as a consequence of adversarial actions, through the perspective of just correct processes. For instance, when a correct process \op{ch-acknowledges} a label~$\ell$, it can either be that a correct process has \op{ch-broadcast} a message~$m$ associated with $\ell$ or that a Byzantine process has injected a message~$m$ associated with $\ell$ in the network. In both cases, the \op{ch-acknowledge} event signals that a message associated with $\ell$ might exist, and therefore it can be used as an anchor for defining potential causality.

We repeat the events characterizing a Byzantine broadcast channel and introduce the new intermediate events. They occur in this order at all correct processes, where $\op{ch-broadcast}$ occurs only at the sender. 

\begin{itemize}
  \item $\op{ch-broadcast}(m)$: a sender process $p_s$ initiates a broadcast for a message~$m$.
  \item $\op{ch-acknowledge}(\ell)$: a process~$p_j \in \CP$ acknowledges a label~$\ell$.
  \item $\op{ch-schedule}(\ell)$: a process~$p_j \in \CP$ schedules a label~$\ell$. 
  \item $\op{ch-deliver}(\ell,m)$: a process~$p_j \in \CP$ delivers a label~$\ell$ with the associated message~$m$.
\end{itemize}

We now use these intermediate events to expose actions in the system from the point of view of a global observer monitoring the correct processes during an execution.
Looking ahead, we want our definition of causality to not rely on the local view at a single process. Instead, we want to be able to substitute local observations with global observations.

The first global event captures the moment when the first correct process receives a label~$\ell$. It signals that a message associated with $\ell$ might exist.

\begin{definition}[Commit event]
  \label{def:commit:event}
A \emph{commit event} for a label~$\ell$, denoted $\CC(\ell)$, is the first point in time when a correct process \op{ch-acknowledges} a label~$\ell$.
\end{definition}

The second global event captures when \emph{enough} correct processes consider a label~$\ell$ ready for delivery. Intuitively, we would like a protocol to postpone the release of payload~$m$ until this moment. 

\begin{definition}[Quorum schedule event]
  \label{def:se:event}
A \emph{quorum schedule event} for a label~$\ell$, denoted $\CQ\CS(\ell)$, is the first point in time when more than $\frac{n-f}{2}$ correct processes have \op{ch-scheduled} a label~$\ell$. 
\end{definition}

Finally, we establish an actual causal-influence relation between these events when occurring for two labels from different senders. Notably, for any label~$\ell$, it trivially holds that $ \CC(\ell)$ precedes $\CQ\CS(\ell)$.

\begin{definition}[Schedule order]
  \label{def:so:event}
  Given two labels $\ell_i$ and $\ell_j$, such that $\ell_i \in \CL_i$ and $\ell_j \in \CL_j$, where $\CL_i \ne \CL_j$, we define a \emph{schedule order} relation, denoted $\ell_i\stackrel{SO}{\to}\ell_j$, whenever $\CQ\CS(\ell_i)$ precedes $ \CC(\ell_j)$. 
\end{definition}

The schedule order is \emph{strict partial order} on labels and associated messages, as it is irreflexive and transitive. 
Intuitively, we would like a protocol preserving causal ordering to postpone the release of payload~$m$ until the moment when the associated label~$\ell$ reaches the quorum schedule event $\CQ\CS(\ell)$. Let $m$ be a message broadcast by a correct process and $\ell$ be the associated label. Then, no event preceding  $\CQ\CS(\ell)$ should depend on the payload $m$.

We are now finally ready to define a new notion of causality in the Byzantine setting. Notably, the relation is defined on the set of messages delivered by correct processes.
Following the structure of Definition~\ref{def:causal:relation}, the new causal order relation covers three cases: (1) messages sent by the same sender, (2) messages sent by different senders, and (3) transitivity. For messages emitted by the same sender, similar to Auvolat~\etal~\cite{DBLP:journals/tcs/AuvolatFRT21}, we require that if a correct process delivers two messages from the same sender, then the order of delivery reflects the order chosen by the sender, independently of the time at which the messages appeared in the system. 
The second case instead captures causal dependencies arising from the global view on protocol events across different processes.
In the crash-fault model, we consider the \emph{local} causal influence relation between the \op{ch-deliver} of a message~$m$ and a \op{ch-broadcast} of a message~$m'$. In the Byzantine setting, we consider the \emph{global} causal influence relation that substitutes the \op{ch-deliver} with an event happening before delivery, and the \op{ch-broadcast} with an event happening after the broadcast in the view of correct processes. This is exactly what the schedule order captures, as it relates the quorum schedule event of a label~$\ell$ with the commit event of another label~$\ell'$. 
The third case is just the transitive closure of the first two. 
We introduce the new notion as follows:

\begin{definition}[Causal-order relation for Byzantine faults]
  \label{def:byz:causal:relation}
Consider two labels $\ell_1$ and $\ell_2$, associated with messages $m_1$ and $m_2$, respectively, delivered by any correct process. We say $(\ell_1, m_1)$ \emph{causally precedes} $(\ell_2, m_2)$ \emph{in the Byzantine setting} and write $(\ell_1, m_1) \prec_B (\ell_2, m_2)$ whenever any of the following conditions hold:
\begin{enumerate}
  \item A correct process $p_i$ delivers $(\ell_1, m_1)$ before $(\ell_2, m_2)$ for the same sender $p_s$, such that $\ell_1, \ell_2 \in \CL_j$; if the sender $p_s$ is correct, it has broadcast $m_1$ before $m_2$.
  \item The two labels $\ell_1$ and $\ell_2$ satisfy the schedule order, i.e., $\ell_1\stackrel{SO}{\to}\ell_2$; or
  \item There exists a label~$\ell'$ associated with message $m'$ such that~$(\ell_1, m_1) \prec_B (\ell', m')$ and $(\ell', m') \prec_B (\ell_2, m_2)$.
\end{enumerate} 
If~$(\ell_1, m_1) \not\prec_B (\ell_2, m_2)$ and~$(\ell_2, m_2) \not\prec_B (\ell_1, m_1)$, we say that the two messages are \emph{concurrent}.
\end{definition}

Since the adversary and the Byzantine processes must not let any information between messages flow \emph{outside} the causal-order relation, we must also add a confidentiality requirement for payloads. In particular, we require that a payload message~$m$ with label~$\ell$ remains hidden up to~$\CQ\CS(\ell)$.

With all necessary ingredients now in place, 
we introduce a new primitive, Byzantine Reliable Broadcast with Causal Ordering (BRB-CO), which provides a secure building block in the Byzantine setting. BRB-CO extends the properties of a Byzantine reliable broadcast channel with three more properties: \emph{completeness}, \emph{hiding}, and \emph{Byzantine causal delivery}. In the following, we define such primitives and repeat the properties of BRCH~(Definition ~\ref{def:brb:channel}) to account for the additional events.

\begin{definition}[Byzantine reliable broadcast with causal ordering (BRB-CO)] 
  \label{def:brb:co}
A \emph{Byzantine reliable broadcast with causal ordering (BRB-CO)} is a Byzantine reliable broadcast channel that satisfies the following properties:
\begin{description}
  \item[\textit{Validity}:] If a correct process $p_s$ broadcasts a message $m$, then every correct process eventually delivers~$m$ with associated label~$\ell$.
  \item[\textit{No duplication}:] For every label~$\ell$, every correct process delivers at most one message with the same label~$\ell$.
  \item[\textit{Integrity}:] If some correct process $p_i$ delivers a message $m$ and label~$\ell$ with $\ell \in \CL_s$ and the sender process $p_s$ is correct, then $m$ was previously broadcast by $p_s$, and $p_i$ had acknowledged label~$\ell$. 
  \item[\textit{Agreement}:] If some correct process delivers a message $m$ with label~$\ell$, then every correct process eventually delivers message $m$ with the same label~$\ell$. 
  \item[\textit{Completeness}:] If some correct process schedules a label~$\ell$, then it will eventually deliver the message $m$ associated to $\ell$.
  \item[\textit{Hiding}:] For any label~$\ell$ associated with a message~$m$ broadcast by a correct process~$p_s$, no process other than~$p_s$ can determine any information about~$m$ before the quorum schedule event~$\CQ\CS(\ell)$ occurs.
  \item[\textit{Byzantine Causal Delivery}:] For any label~$\ell_1$ associated with message $m_1$ that causally precedes $\ell_2$ associated with $m_2$, i.e., $(\ell_1, m_1) \prec_B (\ell_2, m_2)$, no correct process delivers $m_2$ with label~$\ell_2$ unless it has already delivered $m_1$ with label~$\ell_1$.
\end{description}
\end{definition}

Notably, the \emph{completeness} property ties the scheduling of a label to eventual delivery, effectively decoupling agreement on the label from retrieval of the associated message.
The \emph{hiding} property requires any implementation to reveal the payload message after the quorum schedule event, that is, only once enough processes have reliably observed that label.
Together, the two properties, support and guarantee \emph{Byzantine causal delivery}.

Our notion anchors the causal order in events that arise during the execution. This parallels the traditional causal-order definition in the literature on distributed systems and, unlike earlier definitions of causal-order broadcast in the Byzantine setting~\cite{DBLP:conf/crypto/CachinKPS01,DBLP:conf/dsn/DuanRZ17}, it does not define causality directly from cryptographic secrecy.

\section{Implementation of BRB-CO}
\label{sec:protocol}

In this section, we first define a primitive, which we refer to as \emph{verifiable encapsulation (VE)}, which allows to \emph{hide} a message $m$ among a set of $n$ processes, and later to recover it from a fraction of them. We then introduce a complete protocol for BRB-CO. Finally, we study the complexity of the proposed implementation.

\subsection{Verifiable Encapsulation}

We model Verifiable Encapsulation as an \emph{ideal} distributed functionality, to which every process $p_i \in \CP$ has access. The ideal implementation maintains internal shared state through two associative arrays $\CE$ and $\CV$: for each message $m$ the implementation assigns a fresh identifier $\ell_e$ such that $\CE[\ell_e] = m$, and stores a tuple $(\encaps_i, \share_i)$ per process, such that $\CV[p_i, \ell_e] = (\encaps_i, \share_i)$. Notably, we call $\encaps_i$ \emph{encapsulation data} and $\share_i$ \emph{share}, and they are used to validate the encapsulation and recover the data for reconstruction, respectively. Practical instantiation of VE may also merge the roles of $\encaps_i$ and $\share_i$ into one value.

VE provides the following operations: 

\begin{itemize}
    \item $\setup{n, k} \rightarrow pp$. This function initializes the primitive for $n$ processes and reconstruction threshold $k \le n$. It returns public parameters $pp$ to each process $p_i \in \CP$.
    \item $\encapsulate{m} \rightarrow (\ell_e, [\encaps_i]_{i \in [n]})$. On input message $m$ the implementation chooses a fresh encapsulation identifier $\ell_e$ and computes a vector of encapsulated data $[\encaps_i]_{i \in [n]}$ and shares $[\share_i]_{i \in [n]}$  for each process $p_i \in \CP$ in the system. It updates its local state such that $\CE[\ell_e] = m$ and $\CV[p_i, \ell_e] = (\encaps_i, \share_i)$ for each $p_i \in \CP$. It outputs identifier $\ell_e$ and vector $[\encaps_i]_{i \in [n]}$. The vector of shares remains part of the state. Any process may invoke this function.
    \item $\validate{p_i}{\ell_e, \encaps_i} \rightarrow b$. The invocation specifies a caller process~$p_i \in \CP$, and on input an identifier $\ell_e$ and encapsulated data $\encaps_i$, the implementation outputs a boolean value $b \in \{0,1\}$. The response is $1$ if and only if there exist $m$ and $\share_i$ such that $\CE[\ell_e]=m$ and $ \CV[p_i, \ell_e]=(\encaps_i, \share_i)$; otherwise, the response is~$0$. Only $p_i$ may invoke this function. 
    \item $\extract{p_i}{\ell_e, \encaps_i} \to \share_i$, where $s_i$ may also be $\bot$. The invocation specifies a caller process~$p_i \in \CP$, an identifier~$\ell_e$, and encapsulated data~$\encaps_i$. The method outputs a share~$\share_i \neq \bot$ if there exists some $m$ such that $\CE[\ell_e]=m$ and $\CV[p_i, \ell_e]= (\encaps_i, \share_i)$; and $\bot$ otherwise. Only $p_i$ may invoke this function.
    \item $\verify{p_i, \ell_e, \share_i} \rightarrow b$. On input a process $p_i$, an identifier $\ell_e$, and a share $\share_i$, the function outputs a boolean value $b \in \{0,1\}$. The response is $1$ if and only if there exists an $m$ such that $\CE[\ell_e] = m$ and it holds $ \CV[p_i, \ell_e] = (\encaps_i, \share_i)$. Otherwise, the output is $0$. Any process may invoke this function.
    \item $\reconstruct{\ell_e, \CS = \{\share_i\}} \to m$, where $m = \bot$ is also possible. On input an identifier $\ell_e$ and a set of shares $\CS$ from different processes, the function outputs some message~$m$. The output $m$ is not $\bot$ whenever $\CE[\ell_e] = m$ and there exists a set of at least $k$ processes such that for each $p_i$ in this set and $\share_i \in \CS$, it holds $\CV[p_i, \ell_e] = (\encaps_i, \share_i)$ for some $\encaps_i$. Otherwise, the function outputs $\bot$. Any process may invoke this function.
\end{itemize}
Some encapsulation data~$\encaps_i$ or a share $\share_i$ is called \emph{valid} (with respect to label~$\ell_e$ and process~$p_i$) when either $\validate{p_i}{\ell_e, \encaps_i} = 1$ or $\share_i$ is returned by $\extract{p_i}{\ell_e, \encaps_i}$.
According to this specification, the VE primitive provides the following three properties, which are fundamental for our use:
\begin{itemize}
\item \emph{Identifier uniqueness}: Given a identifier $\ell_e$, there exists at most one message $m$ such that $\CE[\ell_e] = m$ and, for every other process $p_i$, there exist at most one tuple $(\encaps_i, \share_i)$ in $\CV[p_i, \ell_e]$. 
\item \emph{Privacy}: Given a set of valid shares $\CS = \{\share_i\}$, i.e., there exist an identifier $\ell_e$ and an encapsulation data $\encaps_i$ such that $\CV[p_i, \ell_e]=(\encaps_i, \share_i)$. If $|\CS| < k$, no information on $m$ is revealed.
\item \emph{Reconstruction}: Given a set of valid shares $\CS = \{\share_i\}$, i.e., there exist an identifier $\ell_e$ and an encapsulation data $\encaps_i$ such that $\CV[p_i, \ell_e] = (\encaps_i, \share_i)$. If $|\CS| \ge k$, the implementation returns a message $m$ stored at $\CE[\ell_e]$. 
\end{itemize}

In Appendix~\ref{app:secure_VE} we discuss how to securely instantiate VE with several cryptographic schemes, including threshold cryptosystems and verifiable secret sharing. Next, we introduce the complete protocol for BRB-CO.

\subsection{Protocol description}

In Alg.~\ref{alg:byz:causal:1}--\ref{alg:byz:causal:2}, we present an implementation of the BRB-CO protocol by relying on the VE primitive to ensure the \emph{hiding} property. The protocol is designed to be resilient to $f$ Byzantine processes, where $n > 3f$.
The protocol follows Bracha's reliable broadcast~\cite{DBLP:journals/iandc/Bracha87}, and extends it to provide causal ordering for a Byzantine reliable broadcast channel (BRCH)~(Definition~\ref{def:brb:channel}).
We characterize the protocol for one payload message by its phases, or communication rounds: the \str{send} phase, the \str{echo} phase, the \str{ready} phase, the \str{schedule} phase, and the \str{reconstruction} phase. Before the \str{echo} phase, a process \op{ch-acknowledges} the payload, and at the end of the \str{ready} phase, a process \op{ch-schedules} the payload according to Section~\ref{sec:causality_definition}.
We refer to the message a correct process sends in a given round with the name of the round, e.g., in the \str{echo} phase, a process sends an \str{echo} message. The communication is one-to-all in the first round, and all-to-all in the others. The protocol implements a Byzantine reliable broadcast channel (BRCH), which handles the broadcast of multiple messages in parallel.

Each message is associated with a label~$\ell$, and the protocol ensures the causal order of the labels, and the associated messages, across all correct processes according to Definition~\ref{def:byz:causal:relation}.
We first present the data structures and their initializations and then give a detailed description of the protocol.

\paragraph{Data structures and initialization.} 
Every process first calls the setup function of VE with the number of processes~$n$ and reconstruction threshold $f < k \le  \lfloor\frac{n-f}{2}\rfloor + 1$. Recall $n>3f$.

The process also maintains several variables. A counter $lsn$ records how many messages the process has locally broadcast. A vector clock $SV$, namely the \emph{scheduled vector}, records the number of labels per sender this process has locally scheduled. Variables $\var{sentecho}$, $\var{sentready}$, $\var{sentschedule}$, $\var{sentrecon}$, and $\var{delivered}$ maintain the respective sets of labels for which the process has locally sent the corresponding protocol message. The variable $\var{encaps}$ stores the encapsulation data $\encaps$ per label, similarly, $\var{echoes}$, $\var{readys}$, $\var{schedules}$, and $\var{undelivered}$ collect the corresponding protocol messages locally received from every other process $p_j$. The temporary sets $\var{pendingecho}$ and $\var{pendingready}$ store not-yet-valid \str{echo} and \str{ready} messages. Finally, $\var{scheduledqueue}$ maintains scheduled but not-yet-delivered messages. Such data structure implements a first-in first-out queue operated though a $\op{enqueue}(v)$ function to add any new element $v$, and a $\op{dequeue}()\rightarrow v$ function to remove the oldest added element. Additionally, $\op{init}()$ initializes the queue as empty and $\op{head}() \rightarrow v$ returns the first element without removing it.

\paragraph{Description of the protocol.}
For every label~$\ell$, the first part of the protocol follows Bracha's reliable broadcast with few modifications ensuring that eventually all correct processes schedule the same label~$\ell$ respecting its potential causal dependencies. Scheduling a label corresponds to delivering the payload message in the original protocol, but BRB-CO adds more phases. The causal dependencies are estimated locally during the execution of the first phases according to the view of multiple processes. To this end, each process locally records in the scheduled vector (SV) how many labels from different senders it has locally scheduled, and sends it along with every new \str{echo} message. The second part allows correct processes to coordinate and jointly recover the associated message $m$.

More precisely, an instance of the protocol consists of the following. A sender process $p_s$ encapsulates $m$ using VE and disseminates it to all processes in a \str{send} message to each process. This contains two values regarding $m$: the encapsulation data~$\encaps_i$, and the associated encapsulation identifier~$\ell_e$. The encapsulation data serves two purposes: it allows every process to validate the identifier~$\ell_e$ before sending an \str{echo} message, and later, if valid, to extract a share~$\share_i$ and to reconstruct the message~$m$. Additionally, \str{send} also contains the locally assigned sequence number~$sn$.

The first modification to Bracha's protocol is that a process validates $\ell_e$ before it sends \str{echo}: if $\ell_e$ is valid, the process assigns to the broadcast instance the label $\ell= (p_s, sn, \ell_e)$, which consists of the sender process $p_s$, the sequence number $sn$, and the encapsulation identifier $\ell_e$. The process also adds a vector clock $W$ to its \str{echo} message, which equals the local value of $SV$. When sending \str{echo} the process also emits the event $\op{brbco-acknowledge}(\ell)$.

The second modification consists of introducing a condition for sending the \str{ready} message; we refer to this condition as \emph{causal barrier}, and it is applied in two places. These two causal barriers prevent a correct process from sending a \str{ready} message for a label~$\ell$ that potentially causally depends on some other label~$\ell'$ before it has locally scheduled $\ell'$. Recall that every label uniquely identifies some associated message $m$, hence, whenever we talk about causal dependencies of labels, we mean causal dependencies of the associated messages.

A correct process sends a \str{ready} message in two different places in the algorithm. The first is upon receiving a Byzantine quorum of \str{echo} messages. The new condition considers an \str{echo} message from a process $p_j$ only valid when the vector clock $W_j$ contained in the message is less than or equal to the local variable~$SV$, i.e., when $W_j \le SV$; this is the \emph{\str{echo} causal barrier}. Its presence implies that the process receiving \str{echo} has locally already scheduled at least the same number of labels per sender as $p_j$ before it proceeds to sending~\str{ready}. Notably, the labels from one correct sender process follow a contiguous, increasing order due to the presence of the sequence number~$sn$ in the label. If an \str{echo} message is not valid yet, it is inserted into $\var{pendingecho}$ until it becomes valid.

The second place where a process can send a \str{ready} message is upon receiving more than $f$ \str{ready} messages at a point in time when it has not yet obtained a Byzantine quorum of \str{echo} messages. A similar condition as in the \str{echo} causal barrier should apply here, and to this effect, every process attaches a vector clock also to its \str{ready} messages.

However, this vector clock cannot be the current variable $SV$ as attached to \str{echo}. When sending \str{ready}, a process should report all potential dependencies, but it is possible that this process is slow and has not yet observed the dependencies a Byzantine quorum in the network has.
Therefore, a process attaches to such a \str{ready} message the \emph{component-wise maximum} $W^{\text{max}}$ of the vectors from the messages that triggered the sending, i.e., either the quorum of \str{echo} messages or the set of more than $f$ \str{ready} messages. 
A process receiving a \str{ready} message from a process $p_k$ considers it valid only after the so-called \str{ready} \emph{causal barrier}, i.e., the vector $W^{\text{max}}_k$ in \str{ready} is less than or equal to the local $SV$.  
If a \str{ready} message is not valid yet, it is inserted into $\var{pendingready}$ until it becomes valid. 

The component-wise maximum ensures that if a label~$\ell$ causally depends on another label $\ell'$, then in every possible Byzantine quorum, there is at least one correct process that sent an \str{echo} message for $\ell$ after scheduling $\ell'$. Consequently, every correct process would send a \str{ready} message for $\ell$ only after scheduling $\ell'$. Moreover, a correct process checks the two causal barriers only until it sends the first \str{ready} message for $\ell$.

Once the process has then received more than $2f$ \str{ready} messages for a label~$\ell$, it proceeds to \emph{scheduling}~$\ell$. This means that it increments $SV$ accordingly (at the index of the process that sent $\ell$), inserts $\ell$ into $\var{scheduledqueue}$, sends a \str{schedule} for~$\ell$, and also emits the $\op{brbco-schedule}(\ell)$ event. This will ensure that the delivery order later follows the scheduling order. 

The last two rounds of the protocol coordinate the processes to extract and release their shares of the message. In particular, a correct process $p_i$, upon receiving more than $\frac{n+f}{2}$ \str{schedule} messages, extracts its share $\share_i$ from $\encaps_i$ (using VE), provided it initially received $\encaps_i$, and broadcasts it to all processes. If $p_i$ has not previously received $\encaps_i$, it simply does not contribute to the reconstruction.
Eventually, every correct process schedules a label~$\ell$ and thus receives more than $\frac{n+f}{2}$ \str{schedule} messages. Therefore, it will receive at least $k$ valid shares from correct processes, which is sufficient to reconstruct the message~$m$ using VE. Recall $f < k \le  \lfloor\frac{n-f}{2}\rfloor + 1$.

The delivery then follows the ordering defined by the scheduling queue.

\begin{algo*}[htbp]
\vbox{
\small
\begin{numbertabbing}\reset
  xxxx\=xxxx\=xxxx\=xxxx\=xxxx\=xxxx\=MMMMMMMMMMMMMMMMMMM\=\kill
  \textbf{State} \label{}\\
  \> $pp \gets \setup{n, k}$ \`//VE functionality setup\label{}\\
  \> $\var{lsn} \gets 0$ \`//local sequence number of broadcast messages\label{}\\
 \> $\var{SV} \gets [0]^n$ \`//vector clock for SCHEDULE messages\label{}\\
  \> $\var{sentecho} \gets \emptyset$\`//keeps the set of labels for which \str{echo} was sent\label{}\\
  \> $\var{sentready} \gets \emptyset$\`//keeps the set of labels for which \str{ready} was sent\label{}\\
  \> $\var{sentschedule} \gets \emptyset$\`//keeps the set of labels for which \str{schedule} was sent\label{}\\
   \> $\var{sentrecon} \gets \emptyset$\`//keeps the set of labels for which \str{reconstruct} was sent\label{}\\
  \> $\var{delivered} \gets \emptyset$\`//keeps the set of labels that have been delivered\label{}\\
  \> $\var{encaps} \gets \emptyset$\`//for each label, collects the valid encapsulation data\label{}\\
  \> $\var{echoes} \gets [\emptyset]^n$\`//collects the valid \str{echo} messages from every other processes\label{}\\
  \> $\var{readys} \gets [\emptyset]^n$\`//collects the valid \str{ready} messages from other processes\label{}\\
  \> $\var{schedules} \gets [\emptyset]^n$\`//collects the valid \str{schedule} messages from other processes\label{}\\
  \> $\var{undelivered} \gets [\emptyset]^n$\`//collects the valid shares from other processes\label{}\\
  \> $\var{pendingecho} \gets \emptyset$\`//collects  \str{echo} messages\label{}\\
  \> $\var{pendingready} \gets \emptyset$\`//collects not-yet-valid \str{ready} messages\label{}\\
  \> $\var{scheduledqueue}.\op{init()}$\`//stores scheduled label in FIFO order\label{}\\
  \\
  \textbf{upon invocation} $\op{brbco-broadcast}(m)$ \textbf{do} \` //only the sender $p_s$\label{}\\
  \> $(\ell_e, [\encaps_1, \dots, \encaps_n]) \gets \encapsulate{m}$ \label{}\\
  \> send message \msg{send}{lsn, \ell_e, \encaps_j} to all $p_j \in \CP$ \label{}\\
  \> $\var{lsn} \gets \var{lsn} + 1$ \label{} \\
  \\
  \textbf{upon} receiving a message \msg{send}{sn, \ell_e, \encaps_i} from $p_s \in \CP$ $\wedge$  $sn \ge 0, \ell_e \not= \bot $\label{}\\
    \textbf{such that} $ \ell = (p_s, sn, \ell_e) \notin \var{sentecho} \wedge \validate{p_i}{\ell_e, \encaps_i}$ \textbf{do}  \\
  \>$\ell \gets (p_s, sn, \ell_e)$ \label{}\\
  \>$\var{sentecho} \gets \var{sentecho} \cup \{ \ell \} $ \label{}\\
  \>$\var{encaps}[\ell] \gets \encaps_i $\label{}\\
   \> \textbf{output} $\op{brbco-acknowledge}(\ell)$ \label{line:acknowledge}\\
  \>send message \msg{echo}{\ell, \var{SV}} to all $p_j \in \CP$ \label{}\\
  \\
  \textbf{upon} receiving a message \msg{echo}{\ell, W} from $p_j$ \textbf{do} \label{}\\
  \> $\var{pendingecho} \gets \var{pendingecho} \cup \{(\ell, p_j, W)\}$ \label{} \\
  \\
  \textbf{upon exists} $(p_j',\ell', W') \in \var{pendingecho}$ 
  \textbf{such that} $W' \le SV$ \textbf{do} \` //\str{echo} causal barrier \label{line:causal_barrier_echo} \\
      \> $\var{pendingecho} \gets \var{pendingecho} \setminus \{(p_j', \ell', W')\}$\label{}\\
  \> $\var{echoes}[ p_j'] \gets (\ell', W')$\label{}\\
    \\
  \textbf{upon exists} $p_s, p_j \in \CP,  sn \geq 0, \ell_e \not= \bot$\label{line:echo_quorum}\\
  \textbf{such that}
  $\ell=(p_s, sn, \ell_e) \wedge
  |\{p_j | \var{echoes}[p_j] = (\ell, \cdot) \} | > \frac{n+f}{2} 
  \wedge \ell \notin \var{sentready}
  $ \textbf{do}
   \\
   \>$\var{sentready} \gets \var{sentready} \cup \{ \ell \} $\label{}\\
  \>$ W_m \gets \op{maxvector}\Bigl(\big[W | (\ell, W) \in \var{echoes}[p_j] \big]\Big)$ \label{}\\
   \> send message \msg{ready}{\ell, W_m} to all $p_j \in \CP$ \label{}\\
  \\
  \textbf{upon} receiving a message \msg{ready}{\ell, W_m} from $p_j$ \textbf{do} \label{}\\
  \> $\var{pendingready} \gets \var{pendingready} \cup \{(p_j, \ell, W_m)\}$ \label{}\\
\\
  \textbf{upon exists} $(p_j', \ell', W_m') \in \var{pendingready}$ 
  \textbf{such that} $W_m' \le \var{SV}$ \textbf{do} \` //\str{ready} causal barrier \label{line:causal_barrier_ready} \\
    \> $\var{pendingready} \gets \var{pendingready} \setminus \{(p_j', \ell', W_m')\}$ \label{}\\
  \> $\var{readys}[p_j'] \gets (\ell', W_m')$ \label{}\\
\\
  \textbf{upon exists} $(p_j', \ell', W_m') \in \var{pendingready}$ 
  \textbf{such that} $\ell' \in \var{sentready}$ \textbf{do} \` //\str{ready} sent \label{} \\
    \> $\var{pendingready} \gets \var{pendingready} \setminus \{(p_j', \ell', W_m')\}$ \` //flush the pending set for $\ell'$ \label{}\\
  \> $\var{readys}[p_j'] \gets (\ell', W_m')$ \label{}\\
\\
\end{numbertabbing}
}
\caption{Byzantine Reliable Broadcast with Causal Ordering (BRB-CO, part 1) for $p_i$}
\label{alg:byz:causal:1}
\end{algo*}

\begin{algo*}[htbp]
\vbox{
\small
\begin{numbertabbing}
  xxxx\=xxxx\=xxxx\=xxxx\=xxxx\=xxxx\=MMMMMMMMMMMMMMMMMMM\=\kill
  \textbf{upon exists} $p_s, p_j \in \CP, sn \geq 0, \ell_e \not= \bot$ \label{line:amplification}\\
  \textbf{such that}
  $\ell=(p_s, sn, \ell_e) \wedge |\{p_j | \var{readys}[p_j] = (\ell, \cdot) \} | > f 
  \wedge \ell \notin \var{sentready}
  $ \textbf{do} \\
  \>$\var{sentready} \gets \var{sentready} \cup \{ \ell \} $ \label{}\\
  \>$ W_m \gets \op{maxvector}\Bigl(\big[W_m' | (\ell, W_m') \in \var{readys}[p_s, sn, j] \big] \Bigr)$ \label{}\\
  \> send message \msg{ready}{\ell, W_m} to all $p_j \in \CP$ \label{}\\
  \\
  \textbf{upon exists} $p_s, p_j \in \CP,  sn \geq 0, \ell_e \not= \bot$ \label{line:ready_quorum}\\
  \textbf{such that}
  $\ell = (p_s, sn, \ell_e) \wedge |\{p_j | \var{readys}[p_j] = (\ell, \cdot) \} | > 2f 
  \wedge \ell \notin \var{sentschedule}
  \wedge sn = \var{SV}[p_s]
  $ \textbf{do}\\
   \>$\var{sentschedule} \gets \var{sentschedule} \cup \{ \ell \} $ \label{}\\
   \> send message \msg{schedule}{\ell} to all $p_j \in \CP$ \label{}\\
    \> $\var{SV}[p_s] \gets \var{SV}[p_s] + 1$ \label{}\\
    \> $\var{scheduledqueue}.\op{enqueue}(\ell)$ \label{line:enqueue}\\ 
   \> \textbf{output} $\op{brbco-schedule}(\ell)$ \label{}\\
\\
  \textbf{upon} receiving a message \msg{schedule}{\ell} from $p_j$ \textbf{do} \label{}\\
  \> $\var{schedules}[p_j] \gets \ell$ \label{}\\
  \\
  \textbf{upon exists} $p_s, p_j \in \CP,  sn \geq 0, \ell_e \not= \bot$ \label{line:schedule_quorum}\\
  \textbf{such that} 
  $\ell = (p_s, sn, \ell_e) \wedge |\{p_j | \var{schedules}[p_j] = \ell \} | > \frac{n+f}{2}
   \wedge \ell \notin \var{sentrecon}$ 
  \textbf{do}\\
  \> $\var{sentrecon} \gets \var{sentrecon} \cup \{ \ell \} $ \label{}\\
  \> \textbf{if} $\var{encaps}[\ell] \not= \bot$ \textbf{then} \label{}\\
  \>\>  $\encaps_i \gets \var{encaps}[\ell]$ \label{}\\
  \>\>  $ \share_i \gets \extract{p_i}{\ell, \encaps_i}$ \label{}\\
  \>\>  send message \msg{reconstruct}{\ell, \share_i} to all $p_j \in \CP$ \label{}\\
  \\
  \textbf{upon} receiving a message \msg{reconstruct}{\ell, \share_j} from $p_j \in \CP \wedge$ \textbf{upon exists} $p_s \in \CP,  sn \geq 0, \ell_e \not= \bot$  \label{}\\
  \textbf{such that} 
  $\ell = (p_s, sn, \ell_e) \wedge \verify{p_j, \ell_e, \share_j}$ \textbf{do} \\
  \> $\var{undelivered}[p_j] \gets (\ell, \share_j) $\label{}\\
\\
\textbf{upon exists} $ p_s, p_j \in \CP, sn \geq 0 , \ell_e \not= \bot$ \textbf{such that}
 $ \ell =(p_s, sn, \ell_e)$ \label{line:share_quorum}\\
 $\wedge |\{p_j | \var{undelivered}[p_j] = (\ell, \cdot) \} | \ge k 
  \wedge \ell \notin \var{delivered} \wedge 
 \var{scheduledqueue}.\op{head()} = \ell$
\textbf{do}\\
  \> $m \gets \reconstruct{\ell, \{s_j | \var{undelivered}[p_j] = (\ell, s_j) \}}$ \label{} \\
  \> $\var{delivered} \gets \var{delivered} \cup \{\ell\}$ \label{}\\
  \> $\ell \gets \var{scheduledqueue}.\op{dequeue()}$ \label{line:dequeue}\\
  \> \textbf{output} $\op{brbco-deliver}(\ell, m)$ \label{}\\
  \\
  \textbf{function} $\op{maxvector(list)}$: \` //where $\var{list}$ is a collection of vector clocks $W$ of size $n$\label{}\\
  \> $\var{result} \gets [0]^n$\label{}\\
  \> \textbf{for} $W \in \var{list}$ \textbf{do} \label{}\\
  \> \> \textbf{for} $i \in [1, n]$ \textbf{do} \label{}\\
  \> \> \> $\var{result}[i] \gets \op{max}(\var{result}[i], W[i])$ \label{}\\
  \> \textbf{return} $\var{result}$\label{}\\
  \\
\end{numbertabbing}
}
\caption{Byzantine Reliable Broadcast with Causal Ordering (BRB-CO, part 2) for $p_i$}
\label{alg:byz:causal:2}
\end{algo*}

The correctness of BRB-CO design depends on all its components. In the following, we informally discuss on the necessity of each one. A secure VE instantiation is fundamental to hide the payload of a message among the set of processes and ensure reconstruction later in the protocol even without the participation of the original sender. Definition~\ref{def:byz:causal:relation} allows us to clearly capture the view of correct processes even in the presence of maliciously injected messages that reach the \emph{schedule} phase and that are later delivered.
Vector clocks shared during the \str{echo} round and incremented during the \str{schedule} round at each process, provide a consistent distributed data structure, with reliable increments, tracking possible causal dependencies. The component-wise maximum calculation on a quorum of vector clocks ensures ``real'' causal dependencies always appear in the local view of a correct process before scheduling a label. 
The causal barrier conditions prevent malicious senders from reporting 
false dependencies in such data structures and provide a wrong snapshot of the system global state. Finally, putting all together, the four all-to-all communication rounds bring the protocol to satisfy causal delivery. 

\subsection{Complexity analysis}
We express the bit complexity of the protocol in terms of the size of the message $|m|$, the number of processes $n$, and $\lambda$, which denotes the maximal size of a unique (cryptographic) label for~$m$.  The exact size may depend on the specific VE instantiation. 
W.l.o.g., we assume the size of the encapsulation data $\encaps$ and the size of a share $\share$ of VE are $O(|m|)$. Future work may explore more efficient VE instantiations that reduce the size of $\encaps$ and $\share$ by leveraging erasure codes and similar techniques.

BRB-CO runs in five communication rounds. In the first round, the sender broadcasts a single message to all other processes; in every subsequent round, each process sends a message to every other process, yielding $O(n^2)$ messages per round and $O(n^2)$ overall, since the number of rounds is constant. The size of the messages exchanged is $O(n|m| + \lambda)$ in the first round, $O(\lambda + n)$ in the second and third rounds due to the inclusion of vector clocks, $O(\lambda)$ in the fourth round, and $O(|m|+\lambda)$ in the fifth round. Combining message and bit complexity across rounds, the overall communication complexity of the protocol is $O(n^2(|m| + \lambda + n))$. Notably, an $\Omega(n)$ lower bound on vector clock size is an unavoidable cost of capturing causality and concurrency information in a distributed system~\cite[Sec.~6.1.3.1]{DBLP:books/daglib/0017536}.

\section{Analysis}
\label{sec:analysis}

In this section, we provide a formal analysis of the proposed protocol and ultimately prove it implements a Byzantine reliable broadcast with causal ordering (BRB-CO). 
First, we prove \emph{agreement} on the label~$\ell$.

\begin{lemma}[Label agreement]
  If a correct process $p_i$ \op{brbco-schedules} a label~$\ell$, then all correct processes eventually \op{brbco-schedule} $\ell$.
  \label{lemma:schedule_agreement}
\end{lemma}

\begin{proof}[Proof by induction]
  (\emph{Base case}): Let us assume all correct processes in the network have not yet \op{brbco-scheduled} any label, i.e., $\var{SV} = [0]^n$ at every correct process. Let us consider a correct process~$p_i$ who \op{brbco-schedules} a label~$\ell$. Following the protocol, $p_i$ received more than $2f$ valid \str{ready} messages for label~$\ell$, and among those messages, more than $f$ are from correct processes. Thus, 
  every correct process~$p_j$ who sent a valid \str{ready} message for $\ell$ attached to it a max vector $W_j$ that is less than or equal to the local $\var{SV}$, i.e., $W_j \le [0]^n$.
  Moreover, for the amplification step (line~\ref{line:amplification}) every correct process~$p_k$, who has not sent \str{ready}, will eventually receive more than $f$ valid \str{ready} messages for label~$\ell$, and send a \str{ready} message for $\ell$ and compute a max vector $W_k \le \var{SV} \le [0]^n$. Hence, every correct process eventually collects more than $2f$ valid \str{ready} messages for label~$\ell$ and thus \op{brbco-schedules}~$\ell$.
  (\emph{Induction step}): Let us assume a set of labels $\Lambda$ such that, if a correct process has \op{brbco-scheduled} all labels $\ell \in \Lambda$, then all correct processes eventually \op{brbco-schedule} all $\ell \in \Lambda$. Therefore, the vector clock~$\var{SV}$ at every correct process is eventually reporting the same number of scheduled labels for each sender; let us denote this vector $\var{SV}_\Lambda$. We want to show that, if a correct process \op{brbco-schedules} a new label $\ell' \notin \Lambda$, then all correct processes eventually schedule $\ell'$. In fact, the same reasoning as in the base case applies: a correct process~$p_i$ scheduling $\ell'$ has received more than $f$ \str{ready} messages from correct processes with max vectors $W_j$ less than or equal to the local $\var{SV} \le \var{SV}_\Lambda$, which implies also $W_j \le \var{SV}_\Lambda$. Again, for the amplification step, every correct process~$p_k$, who has not sent \str{ready}, will eventually receive more than $f$ valid \str{ready} messages for label~$\ell'$, and thus send a \str{ready} message for $\ell'$ and max vector $W_k \le \var{SV}_\Lambda$. Hence, every correct process eventually collects more than $2f$ valid \str{ready} messages for $\ell'$ and thus \op{brbco-schedules}~$\ell'$.
\end{proof}

Given that the agreement property inherently implies reliable updates to the scheduled vector, the next lemma proves that the order in which labels are scheduled at every correct process preserves causal ordering. We first prove the lemma and then describe a possible execution. 

\begin{lemma}[Causal schedule]
  \label{lemma:schedule_causality}
  If $(\ell_1, m_1) \prec_B (\ell_2, m_2)$, then no correct process $\op{brbco-schedules}$ $\ell_2$ before $\ell_1$.
\end{lemma}

\begin{proof}
  From Definition~\ref{def:byz:causal:relation}, we assume that there exists a correct process~$p_i$ that \op{brbco-delivers} $(\ell_1, m_1)$ and $(\ell_2, m_2)$. For each label, we will omit to report the corresponding message in the rest of the proof.
  Following the protocol, $p_i$ has previously \op{brbco-scheduled}~$\ell_1$ and $\ell_2$. Moreover, $p_i$ \op{brbco-delivers} the labels dequeued from the local queue \emph{scheduledqueue}, which implies that $p_i$ has \op{brbco-delivered} either $\ell_1$ before $\ell_2$ or $\ell_2$ before $\ell_1$, and consequently \op{brbco-scheduled} them in the same order.
  We want to show that every correct process \op{brbco-schedules} $\ell_1$ before $\ell_2$. We distinguish between two cases: (1) $\ell_1$ and $\ell_2$ are from the same sender, or (2) from different senders.

  (1) For the same sender the order in the protocol is defined by the sequence numbers. Again, from the first condition in Definition~\ref{def:byz:causal:relation}, we know that there exists a correct process that has \op{brbco-delivered} $\ell_1$ before $\ell_2$. This means that the sequence number attached to $\ell_1$ is smaller than the one attached to $\ell_2$. Following the protocol, the sequence number is part of the label, i.e., $\ell = (p_s, sn, \ell_e)$, and a correct process can only \op{brbco-schedule} a label~$\ell$ with sequence number $sn$ if it has already \op{brbco-scheduled} all labels from the same sender with sequence number smaller than $sn$~(line~\ref{line:ready_quorum}). Hence, by label agreement (Lemma~\ref{lemma:schedule_agreement}), every correct process \op{brbco-schedules} $\ell_1$ before $\ell_2$.

  (2) For different senders the schedule order holds, $\ell_1\stackrel{SO}{\to}\ell_2$, and thus $\text{QS}(\ell_1)$ precedes $\CC(\ell_2)$. 
  The commit event $\CC(\ell_2)$ is defined as the first time a correct process \op{brbco-acknowledges}~$\ell_2$. From the protocol, this happens before the same correct process sends an \str{echo} message for $\ell_2$~(line \ref{line:acknowledge}). 
  Let us consider a correct process $p_k \not= p_i$ who sends a \str{ready} message for $\ell_2$, after receiving more than $\frac{n+f}{2}$ \str{echo} messages for $\ell_2$. A vector clock is attached to each \str{echo}, and among such messages, more than $\frac{n-f}{2}$ are from correct processes. From the schedule order assumption we also know that there are more than $\frac{n-f}{2}$ correct processes that already have \op{brbco-scheduled}~$\ell_1$. Hence, there is at least one correct process $p_j$ that has already \op{brbco-scheduled}~$\ell_1$ and sent an \str{echo} message for $\ell_2$ to $p_k$ with vector clock $W_j$ 
  that has already incremented in the position indexed by the sender of $\ell_1$. 
  More precisely, $W_j[\text{sender}(\ell_1)] \ge \text{sn}(\ell_1) + 1$, where $\text{sender}(\ell_1)$ and $\text{sn}(\ell_1)$ extracts from label~$\ell_1$ the corresponding sender and sequence number, respectively. Moreover, because $p_k$ accepted this \str{echo} from $p_j$, it holds $W_j \le SV_k$, where $SV_k$ is the value of $SV$ at $p_k$. Then $p_k$ must have already \op{brbco-scheduled}~$\ell_1$.
  Notably, this also means that the component-wise maximum vector~$W^{\text{max}}_k$ computed by $p_k$ and attached to \str{ready} will already count $\ell_1$, i.e., $W^\text{max}_k[\text{sender}(\ell_1)] \ge \text{sn}(\ell_1) + 1$, therefore reporting it as a possible dependency.
  We apply the same reasoning to every process sending \str{ready} after a quorum of \str{echo}.
  A similar reasoning holds for a process~$p_q$ sending a \str{ready} in the amplification step, i.e., after receiving more than $f$ \str{ready}. At least one of the \str{ready} messages is sent by a correct process, and, as is clear from the structure of the protocol, and shown formally in many analyses of the Bracha broadcast protocol~\cite{DBLP:books/daglib/0025983}, the first \str{ready} message ever sent by a correct process $p_x$ was sent in response to receiving a quorum of \str{echo}. From before we know that, then $p_x$ will count $\ell_1$ in its \str{ready} message, consequently, if a correct process $p_q$ accepts such a \str{ready} as valid, it means that $W^\text{max}_x \le SV_q$, thus $p_q$ has already \op{brbco-scheduled}~$\ell_1$.
\end{proof}

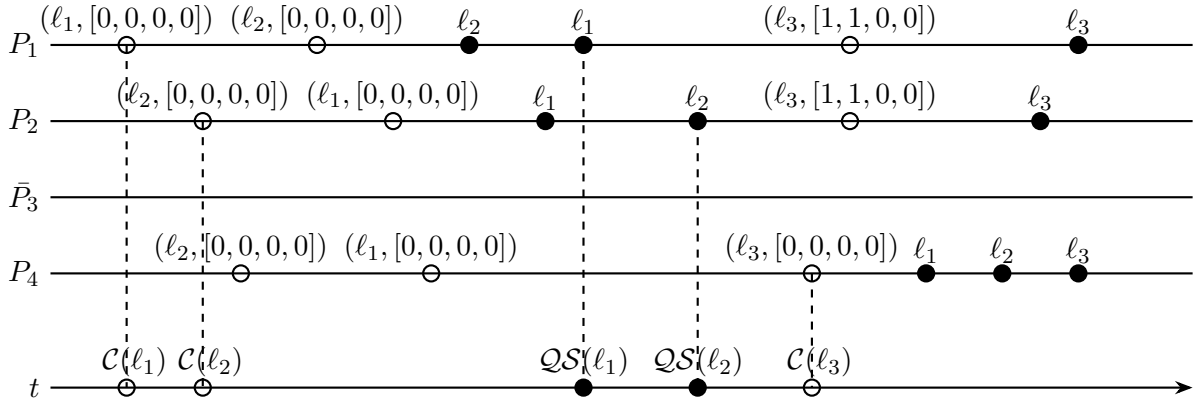
\begin{figure}[hbtp]
  \resizebox{\columnwidth}{!}{
    \centering
  \begin{tikzpicture}[>=Stealth,thick]
    
    \draw (0,0)  node[left]{$P_1$} -- (15,0);
    \draw (0,-1) node[left]{$P_2$} -- (15,-1);
    \draw (0,-2) node[left]{$\bar{P_3}$} -- (15,-2);
    \draw (0,-3) node[left]{$P_4$} -- (15,-3);
    \draw[->] (0,-4.5) node[left]{$t$} -- (15,-4.5);

    \draw[black] (1,0) circle (3pt) node[above]{$(\ell_1, [0,0,0,0])$};
    \draw[black] (3.5,0) circle (3pt) node[above]{$(\ell_2, [0,0,0,0])$};

    \filldraw[black] (5.5,0) circle (3pt) node[above]{$\ell_2$};
    \filldraw[black] (7,0) circle (3pt) node[above]{$\ell_1$};

    \draw[black] (10.5,0) circle (3pt) node[above]{$(\ell_3, [1,1,0,0])$};

    \filldraw[black] (13.5,0) circle (3pt) node[above]{$\ell_3$};

    \draw[black] (2,-1) circle (3pt) node[above]{$(\ell_2, [0,0,0,0])$};
    \draw[black] (4.5,-1) circle (3pt) node[above]{$(\ell_1, [0,0,0,0])$};

    \filldraw[black] (6.5,-1) circle (3pt) node[above]{$\ell_1$};
    \filldraw[black] (8.5,-1) circle (3pt) node[above]{$\ell_2$};

    \draw[black] (10.5,-1) circle (3pt) node[above]{$(\ell_3, [1,1,0,0])$};

    \filldraw[black] (13,-1) circle (3pt) node[above]{$\ell_3$};

    \draw[black] (2.5,-3) circle (3pt) node[above]{$(\ell_2, [0,0,0,0])$};
    \draw[black] (5,-3) circle (3pt) node[above]{$(\ell_1, [0,0,0,0])$};
    \draw[black] (10,-3) circle (3pt) node[above]{$(\ell_3, [0,0,0,0])$};

    \filldraw[black] (11.5,-3) circle (3pt) node[above]{$\ell_1$};
    \filldraw[black] (12.5,-3) circle (3pt) node[above]{$\ell_2$};
    \filldraw[black] (13.5,-3) circle (3pt) node[above]{$\ell_3$};
    
    \draw[black] (1,-4.5) circle (3pt) node[above]{$\;\;\CC(\ell_1)$};
    \draw[black] (2,-4.5) circle (3pt) node[above]{$\;\;\CC(\ell_2)$};
    \filldraw[black] (8.5,-4.5) circle (3pt) node[above]{$\CQ\CS(\ell_2)$};
    \filldraw[black] (7,-4.5) circle (3pt) node[above]{$\CQ\CS(\ell_1)$};
    \draw[black] (10,-4.5) circle (3pt) node[above]{$\;\;\CC(\ell_3)$};
    \draw[-, dashed] (1,0) -- (1,-4.5);
    \draw[-, dashed] (2,-1) -- (2,-4.5);
    \draw[-, dashed] (7,0) -- (7,-4.5);
    \draw[-, dashed] (8.5,-1) -- (8.5,-4.5);
    \draw[-, dashed] (10,-3) -- (10,-4.5);

  \end{tikzpicture}
  }
  \caption{Execution for $\ell_1$, $\ell_2$ and $\ell_3$, where $\ell_1$ and $\ell_2$ are concurrent, and $\ell_3$ causally depends on both $\ell_1$ and $\ell_2$. We show on the temporal line only the global events relevant to observe the causal precedence of $\ell_3$ from $\ell_1$ and $\ell_2$. Notably, $P_4$ that is slightly behind schedules first $\ell_1$ and $\ell_2$, and only after $\ell_3$.}
  \label{fig:exec}
\end{figure}

In Fig.~\ref{fig:exec}, we illustrate an execution of the protocol with four processes, where the adversary controls process~$\bar{P_3}$. W.l.o.g. we represent only $\op{brbco-acknowledge($\ell$)}$ and $\op{brbco-schedule($\ell$)}$ events. The former are denoted by an empty circle, and they occur at every correct process when it sends an \str{echo} message; this message contains the label and the local scheduled vector $\var{SV}$, which we explicitly depict. The latter events, \op{brbco-schedule}, are denoted by a filled circle, they occur at every correct process, and they report just the label. 
As we can observe from the projection on the temporal axis, the messages with labels $\ell_1$ and $\ell_2$ are concurrent, while the one with $\ell_3$ causally depends on both $\ell_1$ and $\ell_2$ following Definition~\ref{def:byz:causal:relation} as both $\CQ\CS(\ell_1)$ and $\CQ\CS(\ell_2)$ precede $\CC(\ell_3)$. Looking at the execution trace, we see that all the correct processes respect the causal relation at delivery. In particular, $P_4$ is the first process emitting $\op{brbco-acknowledge($\ell_3$)}$, however, it will not send a ready for $\ell_3$, until it locally schedules $\ell_1$ and $\ell_2$, reported by both $P_1$ and $P_2$ as possible dependencies, effectively delaying the schedule of $\ell_3$. Notably, even if the faulty $\bar{P_3}$ would inject protocol messages, process~$P_4$ still waits for a Byzantine quorum of \str{echo}, including either the messages of $P_1$ or $P_2$, which include in the attached vector clock the information about $\ell_1$ and $\ell_2$, thus respecting the causal order.

In the following, we show that if a correct process schedules a label~$\ell$, there are enough processes that validated $\ell$, acknowledged it, and are able to reconstruct the associated message $m$.

\begin{lemma}
  \label{lemma:ackn}
  If a correct process $p_i$ \op{brbco-schedules} a label~$\ell$, more than $\frac{n-f}{2}$ correct processes have \op{brbco-acknowledged} the same label~$\ell$.
\end{lemma} 

\begin{proof}
   Let us show that if the number of processes that \op{brbco-acknowledge}~$\ell$ is less than or equal to $\frac{n-f}{2}$, no correct process can \op{brbco-schedule}~$\ell$. Following the protocol, after that a correct process \op{brbco-acknowledges} a label~$\ell$, it sends an \str{echo} message. Let us consider a correct process $p_i$ that receives no more than $\frac{n-f}{2}$ \str{echo} messages from correct processes and other $f$ (potentially from the adversary). But $\frac{n-f}{2} + f$ is strictly less than the necessary quorum to send a \str{ready} message and progress, thus $p_i$ will never move forward. This is true for every correct process. As the protocol requires a correct process to collect more than $2f$ \str{ready} messages in order to \op{brbco-schedule}~$\ell$ and more than $f$ for triggering the amplification step, no correct process can ever \op{brbco-schedule} label~$\ell$.
\end{proof}

\begin{corollary}
  \label{coro:reconstruct}
  If a correct process $p_i$ \op{brbco-schedules} a label~$\ell$, enough correct processes jointly have the necessary information to reconstruct $m$.
\end{corollary}

\begin{proof}
  The corollary directly follows from Lemma~\ref{lemma:ackn} and the protocol: as more than $\frac{n-f}{2}$ correct processes \op{brbco-acknowledge}~$\ell$, given $\ell = (p_s, sn, \ell_e)$, they previously received and validated the encapsulated data $\encaps_i$ received by $p_s$. As the \emph{reconstruction} threshold of VE lives in $f < k \le  \lfloor\frac{n-f}{2}\rfloor + 1$, enough correct processes jointly have the necessary information to reconstruct $m$.
\end{proof}

By using the previous results, we prove that the provided implementation satisfies the BRB-CO properties.

\begin{theorem}
The protocol described in Alg.~\ref{alg:byz:causal:1}--\ref{alg:byz:causal:2} implements the BRB-CO primitive.
\end{theorem}

\begin{proof}
(\emph{Completeness}) Let us consider a correct process $p_i$ that \op{brbco-schedules} label~$\ell$. By Lemma~\ref{lemma:schedule_agreement}, we know that every other correct process eventually \op{brbco-schedules}~$\ell$. Thus, given that there are $n-f$ correct processes in the system, $p_i$ collects enough \str{schedule} messages to trigger reconstruction (line~\ref{line:schedule_quorum}). Similarly, every other process does so. During reconstruction, processes that hold valid $\encaps$ can extract the corresponding share~$\share$. Corollary~\ref{coro:reconstruct} ensures that there are enough correct processes in the system able to reconstruct $m$. Finally, $p_i$ receives such shares~(line~\ref{line:share_quorum}), reconstructs $m$, and delivers $m$ with label~$\ell$. 

(\emph{Validity}) 
Let us assume a correct process~$p_s$ \op{brbco-broadcast} a message~$m$: it encapsulates $m$ and obtains the encapsulation label $\ell_e$ and a vector of encapsulated data $[\encaps_i]_{i \in n}$, thus sends to each process $p_i$ the tuple $(\ell_e, \encaps_i)$ and the local sequence number $sn$. Let us now assume correct process $p_j$ that receives $\encaps_j$ and $\ell_e$ from $p_s$, successfully validates them, assigns the associated label as $\ell =(p_s, sn, \ell_e)$, \op{brbco-acknowledges} it and, finally, sends an \str{echo} message. The \str{echo} message contains also the vector $SV_j$ reporting possible dependencies that $p_j$ has already locally scheduled. Lemma~\ref{lemma:schedule_agreement} ensures that every other correct process eventually schedules those dependencies. The same happens at every other correct process either than $p_j$. Upon receiving \str{echo} messages, $p_j$ validates them against the causal barrier condition~(line~\ref{line:causal_barrier_echo}). Once it collects enough \str{echo}~(line~\ref{line:echo_quorum}), it sends a \str{ready} message. Following the protocol, every correct process sends a \str{ready} message, and upon collecting enough \str{ready}, it \op{brbco-schedules} label~$\ell$. At this point, the rest follows from Lemma~\ref{lemma:schedule_agreement}, namely, every other correct process \op{brbco-schedules} label~$\ell$, and \emph{completeness}. 

(\emph{No duplication} and \emph{Integrity}) The no duplication property and integrity follow from the definition of the label space and the \emph{label uniqueness} property of VE. 

(\emph{Agreement}) Let us consider a correct process $p_i$, who delivers a message $m$ with label~$\ell$. $p_i$ dequeued $\ell$ from the \emph{scheduled queue} (line~\ref{line:dequeue}), thus it must have previously scheduled $\ell$. 
The rest follows from Lemma~\ref{lemma:schedule_agreement}, all other correct processes eventually schedule label~$\ell$, and \emph{completeness}, every process that schedules a label~$\ell$ eventually delivers it.

(\emph{Hiding}) The property directly follows from the algorithm and the \emph{privacy} property of VE. Let us assume a correct process $p_i$ that \op{brbco-schedules} a label~$\ell$. From Corollary~\ref{coro:reconstruct}, we know that enough correct processes jointly have the necessary information to reconstruct $m$. At this point, the adversary knows only $f$ shares, which by the \emph{privacy} property of VE, is not enough to learn information about $m$.
Following the algorithm, a correct process~$p_i$ forwards its share $\share_i$ to all other processes only after collecting more than $\frac{n+f}{2}$ \str{schedule} messages. Because $\frac{n-f}{2}$ of those messages are from correct processes, the adversary may gain information about $m$ only after more than $\frac{n-f}{2}$ correct processes have already \op{brbco-scheduled}~$\ell$, thus effectively after the \emph{quorum schedule event} $\text{QS}(\ell)$.

(\emph{Byzantine Causal Delivery}) The property directly follows from Lemma~\ref{lemma:schedule_causality} and the protocol. The first ensures that, for each correct process, the causal order is preserved at the \op{brbco-schedule}. Then, upon scheduling, labels are inserted in the \emph{scheduled queue}~(line~\ref{line:enqueue}), and \op{brbco-delivered} in the same order in which they are dequeued~(line~\ref{line:dequeue}).

\end{proof}

\section{Related Work}
\label{sec:relwork}

The most prominent methods to impose causal-order delivery in the Byzantine model rely on underlying total-order guarantees together with confidentiality from threshold cryptography.  Early work dates back decades, starting with Reiter and Birman~\cite{DBLP:journals/toplas/ReiterB94,DBLP:conf/crypto/CachinKPS01}. It uses threshold ciphers~\cite{DBLP:journals/joc/ShoupG02} and atomic broadcast tolerating Byzantine faults; later Duan~\etal~\cite{DBLP:conf/dsn/DuanRZ17} propose more efficient implementations based on secret sharing. Like our protocol, these methods require (at least) one extra round of message exchange for each payload message to be delivered. This line of work has recently been revitalized with \emph{batched} threshold encryptions schemes~\cite{DBLP:conf/uss/BormetFOQ25,DBLP:conf/uss/ChoudhuriGP025} that permit to combine the decryption step for multiple messages into one round of message exchange. Malkhi and Szalachowski~\cite{DBLP:conf/tokenomics/MalkhiS22} show how similar techniques seamlessly integrate with DAG-based consensus~\cite{DBLP:conf/podc/KeidarKNS21}.

More recently, Kelkar et al.~\cite{DBLP:conf/crypto/Kelkar0GJ20} have introduced the notion of \emph{fair ordering} or \emph{receive-order fairness} that gives ordering guarantees related to causal order, but completely avoiding any use of encryption. Such a fair order should respect the order in which messages appear on the network according to the view of a majority among the correct processes. The notion is also tied to consensus, i.e., it relies on total-order broadcast. A prominent line of work~\cite{DBLP:conf/osdi/ZhangSCZA20, DBLP:conf/fc/Kursawe21,DBLP:conf/fc/CachinMSZ22,DBLP:conf/ccs/KelkarDLJK23} investigates this property in the context of blockchains. However, many proposed solutions employ strong assumptions, such as synchrony or non-optimal resilience to failures; the resulting properties at the level of message delivery order are weaker than for threshold-cryptography based causal ordering methods.

Misra et al.~\cite{DBLP:conf/icdcn/MisraK23,DBLP:journals/tpds/MisraK24,DBLP:journals/pc/MisraK25} focus on Byzantine-tolerant causal ordering and detection of causality. These works restrict the causality definition to messages sent by correct processes, without trying to capture, as we are doing here, the effects of malicious behavior. This leads to weaker notions of causality. Using similar definitions, Auvolat~\etal~\cite{DBLP:journals/tcs/AuvolatFRT21} show how to design a Byzantine causal order broadcast with implementation and security analysis. Kowalski~\etal~\cite{DBLP:conf/applied/KowalskiMP24} propose a protocol based on \emph{mutual broadcast} that achieves causal ordering, but only when the sender process is correct. 

In the shared-memory model, Teng~\etal~\cite{DBLP:conf/nca/TsengWZP19} propose an algorithm to preserve \emph{causal consistency} in the presence of Byzantine servers, but crash-fault clients. This line of work is complementary to ours: it studies how causality constrains reads and writes on shared objects, whereas we focus on message-passing broadcast primitives. Still, both settings highlight the same core difficulty in the Byzantine model, namely that causality cannot be inferred from the local view of a single participant alone. Shared-memory protocols typically rely on message-passing primitives or quorum-based validation to order operations, which parallels our use of labels, scheduled vectors, and quorum events to recover a global causal relation from local observations.

\section{Conclusion} 
\label{sec:conclusion}

We introduced Byzantine Reliable Broadcast with Causal Ordering (BRB-CO), a broadcast primitive that satisfies a \emph{new} notion of causality in the presence of Byzantine processes. BRB-CO separates causality from total order, mirroring the role causal broadcast plays under crash faults but for the Byzantine setting. We are the first to formally define the causal relation required in this setting, state the properties a suitable broadcast channel must satisfy, and give a protocol based on verifiable encapsulation (VE) and Bracha-style broadcast that realizes them.

Beyond the specific protocol, the main takeaway is that causal ordering under Byzantine faults can be enforced directly at the broadcast layer, without resorting to total order. This opens the door to broadcast-based services that need concurrency while remaining secure against strong adversaries and dependency-manipulation attacks. Natural directions for future work include more efficient instantiations of the encapsulation primitive, applications that exploit BRB-CO as a modular communication substrate for scalable, concurrent Byzantine systems, and a closer study of how this causal definition relates to existing work on secure causal atomic broadcast.

\section*{Acknowledgments}

This work has been supported by a donation from the Ripple Impact Fund in connection with the University Blockchain Research Initiative (UBRI).

\printbibliography

@article{DBLP:journals/toplas/ReiterB94,
  author       = {Michael K. Reiter and
                  Kenneth P. Birman},
  title        = {How to Securely Replicate Services},
  journal      = {{ACM} Trans. Program. Lang. Syst.},
  volume       = {16},
  number       = {3},
  pages        = {986--1009},
  year         = {1994},
  url          = {https://doi.org/10.1145/177492.177745},
  doi          = {10.1145/177492.177745},
  bibsource    = {dblp computer science bibliography, https://dblp.org}
}

@article{DBLP:journals/tcs/AuvolatFRT21,
  author       = {Alex Auvolat and
                  Davide Frey and
                  Michel Raynal and
                  Fran{\c{c}}ois Ta{\"{\i}}ani},
  title        = {Byzantine-tolerant causal broadcast},
  journal      = {Theor. Comput. Sci.},
  volume       = {885},
  pages        = {55--68},
  year         = {2021},
  url          = {https://doi.org/10.1016/j.tcs.2021.06.021},
  doi          = {10.1016/J.TCS.2021.06.021},
  bibsource    = {dblp computer science bibliography, https://dblp.org}
}

@article{DBLP:journals/pc/MisraK25,
  author       = {Anshuman Misra and
                  Ajay D. Kshemkalyani},
  title        = {Byzantine-tolerant detection of causality: There is no holy grail},
  journal      = {Parallel Comput.},
  volume       = {124},
  pages        = {103136},
  year         = {2025},
  url          = {https://doi.org/10.1016/j.parco.2025.103136},
  doi          = {10.1016/J.PARCO.2025.103136},
  bibsource    = {dblp computer science bibliography, https://dblp.org}
}

@article{DBLP:journals/tpds/MisraK24,
  author       = {Anshuman Misra and
                  Ajay D. Kshemkalyani},
  title        = {Byzantine-Tolerant Causal Ordering for Unicasts, Multicasts, and Broadcasts},
  journal      = {{IEEE} Trans. Parallel Distributed Syst.},
  volume       = {35},
  number       = {5},
  pages        = {814--828},
  year         = {2024},
  url          = {https://doi.org/10.1109/TPDS.2024.3368280},
  doi          = {10.1109/TPDS.2024.3368280},
  bibsource    = {dblp computer science bibliography, https://dblp.org}
}

@inproceedings{DBLP:conf/icdcn/MisraK23,
  author       = {Anshuman Misra and
                  Ajay D. Kshemkalyani},
  title        = {Byzantine Fault-Tolerant Causal Ordering},
  booktitle    = {24th International Conference on Distributed Computing and Networking,
                  {ICDCN} 2023, Kharagpur, India, January 4-7, 2023},
  pages        = {100--109},
  publisher    = {{ACM}},
  year         = {2023},
  url          = {https://doi.org/10.1145/3571306.3571395},
  doi          = {10.1145/3571306.3571395},
  bibsource    = {dblp computer science bibliography, https://dblp.org}
}

@inproceedings{DBLP:conf/crypto/CachinKPS01,
  author       = {Christian Cachin and
                  Klaus Kursawe and
                  Frank Petzold and
                  Victor Shoup},
  editor       = {Joe Kilian},
  title        = {Secure and Efficient Asynchronous Broadcast Protocols},
  booktitle    = {Advances in Cryptology - {CRYPTO} 2001, 21st Annual International
                  Cryptology Conference, Santa Barbara, California, USA, August 19-23,
                  2001, Proceedings},
  series       = {Lecture Notes in Computer Science},
  volume       = {2139},
  pages        = {524--541},
  publisher    = {Springer},
  year         = {2001},
  url          = {https://doi.org/10.1007/3-540-44647-8\_31},
  doi          = {10.1007/3-540-44647-8\_31},
  bibsource    = {dblp computer science bibliography, https://dblp.org}
}

@inproceedings{DBLP:conf/dsn/DuanRZ17,
  author       = {Sisi Duan and
                  Michael K. Reiter and
                  Haibin Zhang},
  title        = {Secure Causal Atomic Broadcast, Revisited},
  booktitle    = {47th Annual {IEEE/IFIP} International Conference on Dependable Systems
                  and Networks, {DSN} 2017, Denver, CO, USA, June 26-29, 2017},
  pages        = {61--72},
  publisher    = {{IEEE} Computer Society},
  year         = {2017},
  url          = {https://doi.org/10.1109/DSN.2017.64},
  doi          = {10.1109/DSN.2017.64},
  bibsource    = {dblp computer science bibliography, https://dblp.org}
}

@inproceedings{DBLP:conf/fc/CachinMSZ22,
  author       = {Christian Cachin and
                  Jovana Micic and
                  Nathalie Steinhauer and
                  Luca Zanolini},
  editor       = {Ittay Eyal and
                  Juan A. Garay},
  title        = {Quick Order Fairness},
  booktitle    = {Financial Cryptography and Data Security - 26th International Conference,
                  {FC} 2022, Grenada, May 2-6, 2022, Revised Selected Papers},
  series       = {Lecture Notes in Computer Science},
  volume       = {13411},
  pages        = {316--333},
  publisher    = {Springer},
  year         = {2022},
  url          = {https://doi.org/10.1007/978-3-031-18283-9\_15},
  doi          = {10.1007/978-3-031-18283-9\_15},
  bibsource    = {dblp computer science bibliography, https://dblp.org}
}

@inproceedings{DBLP:conf/ccs/KelkarDLJK23,
  author       = {Mahimna Kelkar and
                  Soubhik Deb and
                  Sishan Long and
                  Ari Juels and
                  Sreeram Kannan},
  editor       = {Weizhi Meng and
                  Christian Damsgaard Jensen and
                  Cas Cremers and
                  Engin Kirda},
  title        = {Themis: Fast, Strong Order-Fairness in Byzantine Consensus},
  booktitle    = {Proceedings of the 2023 {ACM} {SIGSAC} Conference on Computer and
                  Communications Security, {CCS} 2023, Copenhagen, Denmark, November
                  26-30, 2023},
  pages        = {475--489},
  publisher    = {{ACM}},
  year         = {2023},
  url          = {https://doi.org/10.1145/3576915.3616658},
  doi          = {10.1145/3576915.3616658},
  bibsource    = {dblp computer science bibliography, https://dblp.org}
}

@inproceedings{DBLP:conf/osdi/ZhangSCZA20,
  author       = {Yunhao Zhang and
                  Srinath T. V. Setty and
                  Qi Chen and
                  Lidong Zhou and
                  Lorenzo Alvisi},
  title        = {Byzantine Ordered Consensus without Byzantine Oligarchy},
  booktitle    = {14th {USENIX} Symposium on Operating Systems Design and Implementation,
                  {OSDI} 2020, Virtual Event, November 4-6, 2020},
  pages        = {633--649},
  publisher    = {{USENIX} Association},
  year         = {2020},
  url          = {https://www.usenix.org/conference/osdi20/presentation/zhang-yunhao},
  bibsource    = {dblp computer science bibliography, https://dblp.org}
}

@article{DBLP:journals/iandc/Bracha87,
  author       = {Gabriel Bracha},
  title        = {Asynchronous Byzantine Agreement Protocols},
  journal      = {Inf. Comput.},
  volume       = {75},
  number       = {2},
  pages        = {130--143},
  year         = {1987},
  url          = {https://doi.org/10.1016/0890-5401(87)90054-X},
  doi          = {10.1016/0890-5401(87)90054-X},
  bibsource    = {dblp computer science bibliography, https://dblp.org}
}

@inproceedings{DBLP:conf/crypto/Kelkar0GJ20,
  author       = {Mahimna Kelkar and
                  Fan Zhang and
                  Steven Goldfeder and
                  Ari Juels},
  title        = {Order-Fairness for Byzantine Consensus},
  booktitle    = {{CRYPTO} {(3)}},
  series       = {Lecture Notes in Computer Science},
  volume       = {12172},
  pages        = {451--480},
  publisher    = {Springer},
  year         = {2020}
}

@inproceedings{DBLP:conf/tokenomics/MalkhiS22,
  author       = {Dahlia Malkhi and
                  Pawel Szalachowski},
  title        = {Maximal Extractable Value {(MEV)} Protection on a {DAG}},
  booktitle    = {Tokenomics},
  series       = {OASIcs},
  pages        = {6:1--6:17},
  publisher    = {Schloss Dagstuhl - Leibniz-Zentrum f{\"{u}}r Informatik},
  year         = {2022}
}

@article{DBLP:journals/cacm/Lamport78,
  author       = {Leslie Lamport},
  title        = {Time, Clocks, and the Ordering of Events in a Distributed System},
  journal      = {Commun. {ACM}},
  volume       = {21},
  number       = {7},
  pages        = {558--565},
  year         = {1978},
  url          = {https://doi.org/10.1145/359545.359563},
  doi          = {10.1145/359545.359563},
  bibsource    = {dblp computer science bibliography, https://dblp.org}
}

@inproceedings{DBLP:conf/nca/TsengWZP19,
  author       = {Lewis Tseng and
                  Zezhi Wang and
                  Yajie Zhao and
                  Haochen Pan},
  editor       = {Aris Gkoulalas{-}Divanis and
                  Mirco Marchetti and
                  Dimiter R. Avresky},
  title        = {Distributed Causal Memory in the Presence of Byzantine Servers},
  booktitle    = {18th {IEEE} International Symposium on Network Computing and Applications,
                  {NCA} 2019, Cambridge, MA, USA, September 26-28, 2019},
  pages        = {1--8},
  publisher    = {{IEEE}},
  year         = {2019},
  url          = {https://doi.org/10.1109/NCA.2019.8935059},
  doi          = {10.1109/NCA.2019.8935059},
  bibsource    = {dblp computer science bibliography, https://dblp.org}
}

@article{DBLP:journals/joc/ShoupG02,
  author       = {Victor Shoup and
                  Rosario Gennaro},
  title        = {Securing Threshold Cryptosystems against Chosen Ciphertext Attack},
  journal      = {J. Cryptol.},
  volume       = {15},
  number       = {2},
  pages        = {75--96},
  year         = {2002},
  url          = {https://doi.org/10.1007/s00145-001-0020-9},
  doi          = {10.1007/S00145-001-0020-9},
  bibsource    = {dblp computer science bibliography, https://dblp.org}
}

@inproceedings{DBLP:conf/applied/KowalskiMP24,
  author       = {Vincent Kowalski and
                  Achour Most{\'{e}}faoui and
                  Matthieu Perrin},
  editor       = {Ioannis Chatzigiannakis and
                  Vincent Gramoli},
  title        = {Invited Paper: Causal Mutual Byzantine Broadcast},
  booktitle    = {Proceedings of the 2024 Workshop on Advanced Tools, Programming Languages,
                  and PLatforms for Implementing and Evaluating algorithms for Distributed
                  systems, ApPLIED 2024, Nantes, France, 17 June 2024},
  pages        = {1--8},
  publisher    = {{ACM}},
  year         = {2024},
  url          = {https://doi.org/10.1145/3663338.3663679},
  doi          = {10.1145/3663338.3663679},
  bibsource    = {dblp computer science bibliography, https://dblp.org}
}

@book{DBLP:books/daglib/0017536,
  author       = {Hagit Attiya and
                  Jennifer L. Welch},
  title        = {Distributed computing - fundamentals, simulations, and advanced topics
                  {(2.} ed.)},
  series       = {Wiley series on parallel and distributed computing},
  publisher    = {Wiley},
  year         = {2004},
  isbn         = {978-0-471-45324-6},
  bibsource    = {dblp computer science bibliography, https://dblp.org}
}

@book{DBLP:books/daglib/0025983,
  author       = {Christian Cachin and
                  Rachid Guerraoui and
                  Lu{\'{\i}}s E. T. Rodrigues},
  title        = {Introduction to Reliable and Secure Distributed Programming {(2.}
                  ed.)},
  publisher    = {Springer},
  year         = {2011},
  url          = {https://doi.org/10.1007/978-3-642-15260-3},
  doi          = {10.1007/978-3-642-15260-3},
  isbn         = {978-3-642-15259-7},
  bibsource    = {dblp computer science bibliography, https://dblp.org}
}

@inproceedings{DBLP:conf/uss/BormetFOQ25,
  author       = {Jan Bormet and
                  Sebastian Faust and
                  Hussien Othman and
                  Ziyan Qu},
  editor       = {Lujo Bauer and
                  Giancarlo Pellegrino},
  title        = {{BEAT-MEV:} Epochless Approach to Batched Threshold Encryption for
                  {MEV} Prevention},
  booktitle    = {34th {USENIX} Security Symposium, {USENIX} Security 2025, Seattle,
                  WA, USA, August 13-15, 2025},
  pages        = {3457--3476},
  publisher    = {{USENIX} Association},
  year         = {2025},
  url          = {https://www.usenix.org/conference/usenixsecurity25/presentation/bormet},
  bibsource    = {dblp computer science bibliography, https://dblp.org}
}

@inproceedings{DBLP:conf/uss/ChoudhuriGP025,
  author       = {Arka Rai Choudhuri and
                  Sanjam Garg and
                  Guru{-}Vamsi Policharla and
                  Mingyuan Wang},
  editor       = {Lujo Bauer and
                  Giancarlo Pellegrino},
  title        = {Practical Mempool Privacy via One-time Setup Batched Threshold Encryption},
  booktitle    = {34th {USENIX} Security Symposium, {USENIX} Security 2025, Seattle,
                  WA, USA, August 13-15, 2025},
  pages        = {3477--3495},
  publisher    = {{USENIX} Association},
  year         = {2025},
  url          = {https://www.usenix.org/conference/usenixsecurity25/presentation/choudhuri},
  bibsource    = {dblp computer science bibliography, https://dblp.org}
}

@inproceedings{DBLP:conf/podc/KeidarKNS21,
  author       = {Idit Keidar and
                  Eleftherios Kokoris{-}Kogias and
                  Oded Naor and
                  Alexander Spiegelman},
  editor       = {Avery Miller and
                  Keren Censor{-}Hillel and
                  Janne H. Korhonen},
  title        = {All You Need is {DAG}},
  booktitle    = {{PODC} '21: {ACM} Symposium on Principles of Distributed Computing,
                  Virtual Event, Italy, July 26-30, 2021},
  pages        = {165--175},
  publisher    = {{ACM}},
  year         = {2021},
  url          = {https://doi.org/10.1145/3465084.3467905},
  doi          = {10.1145/3465084.3467905},
  bibsource    = {dblp computer science bibliography, https://dblp.org}
}

@inproceedings{DBLP:conf/fc/Kursawe21,
  author       = {Klaus Kursawe},
  editor       = {Matthew Bernhard and
                  Andrea Bracciali and
                  Lewis Gudgeon and
                  Thomas Haines and
                  Ariah Klages{-}Mundt and
                  Shin'ichiro Matsuo and
                  Daniel Perez and
                  Massimiliano Sala and
                  Sam Werner},
  title        = {Wendy Grows Up: More Order Fairness},
  booktitle    = {Financial Cryptography and Data Security. {FC} 2021 International
                  Workshops - CoDecFin, DeFi, VOTING, and WTSC, Virtual Event, March
                  5, 2021, Revised Selected Papers},
  series       = {Lecture Notes in Computer Science},
  pages        = {191--196},
  publisher    = {Springer},
  year         = {2021},
  url          = {https://doi.org/10.1007/978-3-662-63958-0\_17},
  doi          = {10.1007/978-3-662-63958-0\_17},
  bibsource    = {dblp computer science bibliography, https://dblp.org}
}

@article{DBLP:journals/dc/GuerraouiKMPS22,
  author       = {Rachid Guerraoui and
                  Petr Kuznetsov and
                  Matteo Monti and
                  Matej Pavlovic and
                  Dragos{-}Adrian Seredinschi},
  title        = {The consensus number of a cryptocurrency},
  journal      = {Distributed Comput.},
  volume       = {35},
  number       = {1},
  pages        = {1--15},
  year         = {2022},
  url          = {https://doi.org/10.1007/s00446-021-00399-2},
  doi          = {10.1007/S00446-021-00399-2},
  bibsource    = {dblp computer science bibliography, https://dblp.org}
}

@article{DBLP:journals/dc/MalkhiR98,
  author    = {Dahlia Malkhi and
               Michael K. Reiter},
  title     = {Byzantine Quorum Systems},
  journal   = {Distributed Comput.},
  volume    = {11},
  number    = {4},
  pages     = {203--213},
  year      = {1998}
}

@article{DBLP:journals/csur/ChocklerKV01,
  author       = {Gregory V. Chockler and
                  Idit Keidar and
                  Roman Vitenberg},
  title        = {Group communication specifications: a comprehensive study},
  journal      = {{ACM} Comput. Surv.},
  volume       = {33},
  number       = {4},
  pages        = {427--469},
  year         = {2001},
  url          = {https://doi.org/10.1145/503112.503113},
  doi          = {10.1145/503112.503113},
  bibsource    = {dblp computer science bibliography, https://dblp.org}
}

@inproceedings{DBLP:conf/uss/TorresCS21,
  author       = {Christof Ferreira Torres and
                  Ramiro Camino and
                  Radu State},
  editor       = {Michael D. Bailey and
                  Rachel Greenstadt},
  title        = {Frontrunner Jones and the Raiders of the Dark Forest: An Empirical
                  Study of Frontrunning on the Ethereum Blockchain},
  booktitle    = {30th {USENIX} Security Symposium, {USENIX} Security 2021, August 11-13,
                  2021},
  pages        = {1343--1359},
  publisher    = {{USENIX} Association},
  year         = {2021},
  url          = {https://www.usenix.org/conference/usenixsecurity21/presentation/torres},
  bibsource    = {dblp computer science bibliography, https://dblp.org}
}

@article{DBLP:journals/corr/abs-2501-06531,
  author       = {Srivatsan Sridhar and
                  Alberto Sonnino and
                  Lefteris Kokoris{-}Kogias},
  title        = {Stingray: Fast Concurrent Transactions Without Consensus},
  journal      = {CoRR},
  volume       = {abs/2501.06531},
  year         = {2025},
  url          = {https://doi.org/10.48550/arXiv.2501.06531},
  doi          = {10.48550/ARXIV.2501.06531},
  eprinttype   = {arXiv},
  eprint       = {2501.06531},
  bibsource    = {dblp computer science bibliography, https://dblp.org}
}

@inproceedings{DBLP:conf/dsn/CollinsGKKMPPST20,
  author       = {Daniel Collins and
                  Rachid Guerraoui and
                  Jovan Komatovic and
                  Petr Kuznetsov and
                  Matteo Monti and
                  Matej Pavlovic and
                  Yvonne{-}Anne Pignolet and
                  Dragos{-}Adrian Seredinschi and
                  Andrei Tonkikh and
                  Athanasios Xygkis},
  title        = {Online Payments by Merely Broadcasting Messages},
  booktitle    = {50th Annual {IEEE/IFIP} International Conference on Dependable Systems
                  and Networks, {DSN} 2020, Valencia, Spain, June 29 - July 2, 2020},
  pages        = {26--38},
  publisher    = {{IEEE}},
  year         = {2020},
  url          = {https://doi.org/10.1109/DSN48063.2020.00023},
  doi          = {10.1109/DSN48063.2020.00023},
  bibsource    = {dblp computer science bibliography, https://dblp.org}
}

@inproceedings{DBLP:conf/sp/DaianGKLZBBJ20,
  author       = {Philip Daian and
                  Steven Goldfeder and
                  Tyler Kell and
                  Yunqi Li and
                  Xueyuan Zhao and
                  Iddo Bentov and
                  Lorenz Breidenbach and
                  Ari Juels},
  title        = {Flash Boys 2.0: Frontrunning in Decentralized Exchanges, Miner Extractable
                  Value, and Consensus Instability},
  booktitle    = {2020 {IEEE} Symposium on Security and Privacy, {SP} 2020, San Francisco,
                  CA, USA, May 18-21, 2020},
  pages        = {910--927},
  publisher    = {{IEEE}},
  year         = {2020},
  url          = {https://doi.org/10.1109/SP40000.2020.00040},
  doi          = {10.1109/SP40000.2020.00040},
  bibsource    = {dblp computer science bibliography, https://dblp.org}
}

@inproceedings{DBLP:conf/icdcs/AlposCMZ21,
  author       = {Orestis Alpos and
                  Christian Cachin and
                  Giorgia Azzurra Marson and
                  Luca Zanolini},
  title        = {On the Synchronization Power of Token Smart Contracts},
  booktitle    = {41st {IEEE} International Conference on Distributed Computing Systems,
                  {ICDCS} 2021, Washington DC, USA, July 7-10, 2021},
  pages        = {640--651},
  publisher    = {{IEEE}},
  year         = {2021},
  url          = {https://doi.org/10.1109/ICDCS51616.2021.00067},
  doi          = {10.1109/ICDCS51616.2021.00067},
  bibsource    = {dblp computer science bibliography, https://dblp.org}
}

@inproceedings{DBLP:conf/sosp/BirmanJ87,
  author       = {Kenneth P. Birman and
                  Thomas A. Joseph},
  editor       = {Les Belady},
  title        = {Exploiting Virtual Synchrony in Distributed Systems},
  booktitle    = {Proceedings of the Eleventh {ACM} Symposium on Operating System Principles,
                  {SOSP} 1987, Stouffer Austin Hotel, Austin, Texas, USA, November 8-11,
                  1987},
  pages        = {123--138},
  publisher    = {{ACM}},
  year         = {1987},
  url          = {https://doi.org/10.1145/41457.37515},
  doi          = {10.1145/41457.37515},
  bibsource    = {dblp computer science bibliography, https://dblp.org}
}

@inproceedings{DBLP:conf/opodis/RyabininGS25,
  author       = {Fedor Ryabinin and
                  Alexey Gotsman and
                  Pierre Sutra},
  editor       = {Andrei Arusoaie and
                  Emanuel Onica and
                  Michael Spear and
                  Sara Tucci Piergiovanni},
  title        = {Making Democracy Work: Fixing and Simplifying Egalitarian Paxos},
  booktitle    = {29th International Conference on Principles of Distributed Systems,
                  {OPODIS} 2025, Ia{\c{s}}i, Romania, December 3-5, 2025},
  series       = {LIPIcs},
  volume       = {361},
  pages        = {22:1--22:19},
  publisher    = {Schloss Dagstuhl - Leibniz-Zentrum f{\"{u}}r Informatik},
  year         = {2025},
  url          = {https://doi.org/10.4230/LIPIcs.OPODIS.2025.22},
  doi          = {10.4230/LIPICS.OPODIS.2025.22},
  bibsource    = {dblp computer science bibliography, https://dblp.org}
}

@inproceedings{DBLP:conf/sosp/MoraruAK13,
  author       = {Iulian Moraru and
                  David G. Andersen and
                  Michael Kaminsky},
  editor       = {Michael Kaminsky and
                  Mike Dahlin},
  title        = {There is more consensus in Egalitarian parliaments},
  booktitle    = {{ACM} {SIGOPS} 24th Symposium on Operating Systems Principles, {SOSP}
                  '13, Farmington, PA, USA, November 3-6, 2013},
  pages        = {358--372},
  publisher    = {{ACM}},
  year         = {2013},
  url          = {https://doi.org/10.1145/2517349.2517350},
  doi          = {10.1145/2517349.2517350},
  bibsource    = {dblp computer science bibliography, https://dblp.org}
}

@article{DBLP:journals/corr/abs-2506-01885,
  author       = {Atefeh Zareh Chahoki and
                  Maurice Herlihy and
                  Marco Roveri},
  title        = {SoK: Concurrency in Blockchain - {A} Systematic Literature Review
                  and the Unveiling of a Misconception},
  journal      = {CoRR},
  volume       = {abs/2506.01885},
  year         = {2025},
  url          = {https://doi.org/10.48550/arXiv.2506.01885},
  doi          = {10.48550/ARXIV.2506.01885},
  eprinttype   = {arXiv},
  eprint       = {2506.01885},
  bibsource    = {dblp computer science bibliography, https://dblp.org}
}

@article{DBLP:journals/dc/DickersonGHK20,
  author       = {Thomas D. Dickerson and
                  Paul Gazzillo and
                  Maurice Herlihy and
                  Eric Koskinen},
  title        = {Adding concurrency to smart contracts},
  journal      = {Distributed Comput.},
  volume       = {33},
  number       = {3-4},
  pages        = {209--225},
  year         = {2020},
  url          = {https://doi.org/10.1007/s00446-019-00357-z},
  doi          = {10.1007/S00446-019-00357-Z},
  bibsource    = {dblp computer science bibliography, https://dblp.org}
}

@inproceedings{DBLP:conf/eurosys/LinFZ025,
  author       = {Haoran Lin and
                  Hang Feng and
                  Yajin Zhou and
                  Lei Wu},
  title        = {ParallelEVM: Operation-Level Concurrent Transaction Execution for
                  EVM-Compatible Blockchains},
  booktitle    = {Proceedings of the Twentieth European Conference on Computer Systems,
                  EuroSys 2025, Rotterdam, The Netherlands, 30 March 2025 - 3 April
                  2025},
  pages        = {211--225},
  publisher    = {{ACM}},
  year         = {2025},
  url          = {https://doi.org/10.1145/3689031.3696063},
  doi          = {10.1145/3689031.3696063},
  bibsource    = {dblp computer science bibliography, https://dblp.org}
}

@incollection{HadzilacosT93,
  author        = {Vassos Hadzilacos and Sam Toueg},
  title         = {Fault-Tolerant Broadcasts and Related Problems},
  editor        = {Sape J. Mullender},
  booktitle     = {Distributed Systems (2nd Ed.)},
  publisher     = {ACM Press \& Addison-Wesley},
  year          = 1993,
  address       = {New York},
  note          = {Expanded version appears as Technical Report TR94-1425,
                  Department of Computer Science, Cornell University, Ithaca
                  NY, 1994.}
}

\appendix

\section{Secure VE instantiations}
\label{app:secure_VE}

Instantiating the verifiable encapsulation functionality requires cryptographic mechanisms. In particular, VE must: (i) hide the content of a message before sufficiently many processes trigger release (\emph{threshold privacy}), (ii) bind one identifier to one cryptographic transcript and therefore to one message (\emph{binding}), (iii) allow processes to prove that released shares are valid (\emph{verifiability}), and (iv) guarantee that sufficiently many valid shares reconstruct exactly the same message (\emph{correctness} and \emph{robustness}). These properties prevent an adversary from learning or manipulating the content too early. We discuss three concrete instantiations and their trade-offs.

\begin{itemize}

\item A threshold cipher provides confidentiality under a public key shared by the set of processes. The corresponding secret decryption key is not held by one process; it is split into key shares distributed among processes. Reconstruction is possible only when sufficiently many processes provide valid partial decryptions.
In this setting, $\encapsulate{m}$ encrypts $m$ and produces a ciphertext $c$. For every process $p_i$, the encapsulated value is the same, i.e., $\encaps_i = c$. The identifier is defined as $\ell_e = H(c)$, where $H$ is a collision-resistant hash function. Function $\validate{p_i}{\ell_e,\encaps_i}$ checks that $\encaps_i$ is a well-formed ciphertext and that its hash matches the label. Then, $\extract{p_i}{\ell_e,\encaps_i}$ computes a partial decryption share using the secret key share of $p_i$, and $\verify{p_i,\ell_e,\share_i}$ verifies the partial decryption share from process $p_i$. Finally, $\reconstruct{\ell_e,\CS}$ combines sufficiently many valid shares to recover $m$. This instantiation naturally works over authenticated point-to-point channels at the cost of key setup through either a trusted dealer or distributed key generation.

\item Verifiable secret sharing (VSS) guarantees consistency of shares of one secret distributed among the set of processes, allowing each process to verify that its share corresponds to one committed secret.
In this setting, $\encapsulate{m}$ secret-shares $m$ into multiple shares and sends an encapsulated value $\encaps_i$ privately to each process $p_i$, together with a commitment $\ell_e$ binding all shares to the same secret. Validation through $\validate{p_i}{\ell_e,\encaps_i}$ checks the local share against that commitment. Extraction with $\extract{p_i}{\ell_e,\encaps_i}$ corresponds to releasing the already-held share, and $\verify{p_i,\ell_e,\share_i}$ checks that a share released by another process is consistent with the commitment. Finally, $\reconstruct{\ell_e,\CS}$ recovers $m$ from sufficiently many valid shares. In this instantiation, confidentiality additionally relies on secure point-to-point channels during share distribution.

\item Publicly verifiable secret sharing (PVSS) provides publicly checkable correctness of the sharing process, so that anyone can verify that revealed shares are consistent with one public transcript. In this setting, $\encapsulate{m}$ runs PVSS sharing on $m$ and outputs a public transcript $T$ and one designated encrypted share per process. The label is derived from the transcript, e.g., $\ell_e = H(T)$. For each process $p_i$, $\encaps_i$ is the designated encrypted share for $p_i$ in $T$. Similar to VSS, $\validate{p_i}{\ell_e,\encaps_i}$ checks that the designated encrypted share is well-formed with respect to $T$ and that $\ell = H(T)$. Then, $\extract{p_i}{\ell_e,\encaps_i}$ means that $p_i$ decrypts its own share. Verification with $\verify{p_i,\ell_e,\share_i}$ publicly checks that the released share of $p_i$ is valid under $T$. Finally, $\reconstruct{\ell_e,\CS}$ combines sufficiently many valid released shares to recover $m$. This instantiation does not require secure channels, as the public transcript ensures that shares are consistent and valid. However, PVSS schemes are generally more complex and less efficient than VSS or threshold ciphers, both in terms of communication and computation.

\end{itemize}

Threshold ciphers are attractive when confidentiality over public channels is needed and the assumption of a trusted setup is not a limitation. VSS with secure channels is convenient when private channels are already available and a simpler sharing-based construction is preferred. PVSS removes the need for secure channels while preserving public verifiability, at the cost of higher communication and computational complexity than plain VSS.

\end{document}